\documentclass[10pt, conference, letterpaper]{IEEEtran}
\makeatletter
\def\ps@headings{%
\def\@oddhead{\mbox{}\scriptsize\rightmark \hfil \thepage}%
\def\@evenhead{\scriptsize\thepage \hfil \leftmark\mbox{}}%
\def\@oddfoot{}%
\def\@evenfoot{}}
\makeatother
\usepackage{url}
\usepackage{amsmath, amssymb}
\usepackage{latexsym}

\usepackage{times}
\usepackage{mystyle}

\usepackage{xspace}
\usepackage{marginfigs}
\usepackage{graphicx,color}
\usepackage{cite}

\usepackage{epsfig}
\usepackage{url}
\usepackage{fancybox}
\usepackage{pseudocode}
\usepackage{multirow}
\usepackage{wrapfig}

\usepackage{subfig}

\newcommand{\denselist}{\itemsep 0pt\parsep=1pt\partopsep 0pt}
\newcommand{\bitem}{\begin{itemize}\denselist}
\newcommand{\eitem}{\end{itemize}}
\newcommand{\benum}{\begin{enumerate}\denselist}
\newcommand{\eenum}{\end{enumerate}}

\def\note#1{$\langle\langle${\bf #1}$\rangle\rangle$\message{Remove note before final version!}}
\long\def\pa#1{} 

\begin{document}
\title{\huge Space Filling Curves for 3D Sensor Networks \\with Complex Topology}

\author{
	\authorblockN{
		Mayank Goswami\authorrefmark{2} \hspace*{1.cm} Siming Li\authorrefmark{1} \hspace*{1.cm} Junwei Zhang\authorrefmark{1}\\  Emil Saucan\authorrefmark{3} \hspace*{1.cm} Xianfeng Gu\authorrefmark{1}  \hspace*{1.cm} Jie Gao\authorrefmark{1}\\ }
	\vspace*{.3cm}
	\authorblockA{
		\small
		\authorrefmark{1}
		Department of Computer Science, Stony Brook University. \{simingli, gu, jgao\}@cs.stonybrook.edu
		\\
		\authorrefmark{2}
		Max-Planck Institute for Informatics, Germany.
		gmayank@mpi-inf.mpg.de
		\\
		\authorrefmark{2}
		Technion, Israel Institute of Technology; Max Planck Institute for Mathematics, Leipzig, Germany. semil@ee.technion.ac.il
	}
}

\maketitle

\begin{abstract}

Several aspects of managing a sensor network (e.g., motion planning for data mules, serial data fusion and inference) benefit once the network is linearized to a path. 
The linearization is often achieved by constructing a space filling curve inside the domain. 
But existing methods cannot handle networks distributed on surfaces of complex topology.

This paper presents a novel method for generating space filling curves to $3$D sensor networks that are distributed densely on some two-dimensional geometric surface. Our algorithm is completely distributed and constructs a path which gets uniformly, progressively denser as the path becomes longer. We analyze the algorithm mathematically and prove that the curve we obtain is dense. Our method is based on Hodge Decomposition theorem and uses holomorphic differentials on Riemann surfaces. The underlying high genus surface is conformally mapped to a union of flat tori and then a proportionally-dense space filling curve on this union is constructed. The pullback of this curve to the original network gives us the desired curve.
We show via simulations that our method handles complicated topologies and performs much better than the alternative methods on both density and coverage rate. 
\end{abstract}

\section{Introduction}

In this paper we consider sensors deployed in $3$D space such that the sensors are located densely on some underlying $2$-dimensional geometric surface of possibly complex topology. This assumption models many practical scenarios in sensor deployment --- sensors are often attached to the surfaces of terrains, exterior/interior of buildings~\cite{Klann09tactical},  or other architectural structures, for easy installation and energy supplies, etc. In some other cases, the applications require sensors to be installed to monitor complex 3D structures, such as underground tunnels~\cite{Li09underground,Wu10holistic} or pipes~\cite{sokullu11data}. Therefore the sensors are located sparsely in 3D space but densely on a 2-dimensional surface (the ``boundary'' of some $3$D objects) of possibly complex topology. 

We are interested in innovative ways of managing such sensor networks, in the regime of using mobile entities to aid such management. Such mobile agents are often termed `data mules', since one of the major applications is to use a mobile node to collect data from static sensors~\cite{Jea:2005:MCM,somasundara06controllably,Vasilescu:2005:DCS,mascarenas08demonstration,anastasi08data}. Static data sinks suffer from the well known problem of `energy hole', as sensors near the sink are used more often and may run out of battery sooner than others. Mobile data sinks could get around this problem. Other applications of data mules include battery recharge, event response, as well as generic network health monitoring and repairment. Data mules may piggyback on entities with natural mobility such as human, animals or vehicles, but can also be dedicated mobile nodes with either active (i.e., controlled motion) or passive motion (e.g., carried away by wind or water current). The two-layer architecture with static sensors ensuring high spatial coverage and resolution and with mobile nodes allowing flexibity and control thus created a lot of interest in recent years. See the survey paper for a summary~\cite{Ekici06mobility}. 

When the mobile nodes have active motion and are dedicated for improving the effiency of the sensor networks, a central problem is the planning of its motion~\cite{ma07sencar,Kansal:2004:IFI,somasundara06controllably,xing07rendezvous,somansundara04mobile,zhao03message}. There are basically two major design principles. The first one is mainly for handling dynamic, spontaneous events and the data mules often behave in an on-demand manner. Metrics to optimize often include the response delay and total distance travelled. When the main objective is to minimize response delay, the proposed solutions in the literature are often variations of vehicular routing problem~\cite{arun07mobile}, which are known to be NP-hard~\cite{VRP-book}. When the main objective is to minimize the total distance travelled, the problem boils down to various variants of traveling salesman problem~\cite{Applegate:2007:TSP:1374811}, again, an NP-hard one~\cite{Garey1990computers}. When multiple data mules are in place, the coordination of them becomes much harder. Multi-TSP problem is NP-hard and does not have good approximations. 

The second design principle is for a sensor network with dense sensor distribution and the target is to plan data mule along a path that uniformly traverses the entire field. This represents a periodic solution by which all sensors have fair chance of being served by the data mule. It is better for scenarios when sensors have the same data generation rate, or when the sensor network requires a patrolling team to continuously monitor its general functioning and health. It can also be used for linear, logical operations in sensor networks such as data fusion~\cite{Patil04serialdata} or sequential inference. In contrast to the data mule planning problem as listed above, this problem can be abstracted as `linearization' of the sensor field which can substantially benefit from the underlying geometric embedding. 

The representative work in this direction is by Ban~\etal~\cite{ban13topology}, in which they generalize the idea of a space filling curve, often defined for a square, for a general 2D domain with holes.  A space filling curve is a single curve that recursively `fills up' the square, when the number of iterations goes to infinity~\cite{sagan94space}. For a sensor network with fixed density, a space filling curve nicely tours around all sensors with total travel length comparable to the traveling salesman solution. When the sensor network has holes, however, the space filling curve is broken and loses its nice properties. For a domain with a single hole, Ban~\etal~\cite{ban13topology} proposed to map it to a torus such that a space filling curve can be easily found --- by essentially following a line bouncing back and forth between the inner and outer boundaries. When there are more than one holes, all but one holes are mapped to `slits', and the path bounces on these slits too. It is rigorously proved that this curve is dense, i.e., any point of the domain will be covered by the path of sufficiently long; and the curve has bounded density, i.e., it does not path through any point too many times. 

This space filling curve idea for data mule planning has a number of advantages. First the density or coverage of the curve is uniformly, steadily increasing with the length of the path. This is one major difference from the traditional space filling curves in squares, which visits all points in one quadrant before moving to the next. The curve defined above takes a coarse sample of the sensor field and increases density of coverage progressively as the mule travels more. It is easy to tailor the trip for a given travel budget. The second advantage is that it allows multiple data mules to work with each other easily. All of them can follow their respective space filling curves, starting from different positions. These space filling curves in theory do not overlap each other, and in practice naturally complement each other. 

However, one major limitation of the mechanism in Ban~\etal~\cite{ban13topology} is its applicability to 2D domains with holes only. Terrains with holes can be handled via an additional mapping to 2D but general surfaces with high genus (multiple handles) cannot be handled. For underground deployment of sensors for monitoring tunnels, handles appear very often and we need a different scheme for generating the space filling curves. 

\smallskip\noindent\textbf{Our contribution.} The main result in this paper is a new linearization scheme for general sensor networks on 2D surfaces. We generate space filling curves with the same nice properties as those in~\cite{ban13topology}. In particular, the curves have 1) dense, progressive coverage -- that as the curve gets longer, the distance from any point to the curve decreases quickly; 2) uniform density -- a point is not visited more than a constant number of times. 

Moving from a 2D domain to a general 2D surface in 3D really makes the problem harder. New ideas are needed to make it work. The construction presented in this paper is new for mathematicians too. It has been known that dense curve exists but nobody knows how to compute one (even in the centralized way). Therefore the contribution is not only on the algorithm and application aspect, but also in the theoretical perspective. The algorithms presented here can be easily made to work in a distributed setting for a sensor network. We remark that the problem will be much easier if one drops the progressive density requirement. For that one can cut the 2D surface into small patches each mapped to a 2D domain such that previous methods can be applied.  

\section{Related Work}

In this section we quickly survey other ideas that generate a path to visit all sensor nodes. 

\noindent\textbf{Space filling curves.} 
Space filling curves have been used for linear/serial fusion~\cite{Patil04serialdata} in sensor networks when the sensors are deployed uniformly in a square. There has been a heuristic algorithm that generalizes a Hilbert curve for an ellipse~\cite{kamat07hilbert}. Technically the curves generated by Ban~\etal~\cite{ban13topology} and the one in this paper are not going to completely fill up the surface (since topologically a curve is different from a surface) -- but both the curves get infinitely denser as their lengths go to infinity. Thus the curve gets infinitesimally closer to every point.

\smallskip\noindent\textbf{Finding a tour.} On a sensor network deployed in space, generating a tour of the graph, depending on the requirement, maps to either the Hamiltonian cycle/path problem or the traveling salesman tour. The former requires each vertex be visited exactly once and only the edges of the graph can be used. The second tries to minimizes the total travel distance instead. Both problems are NP-hard~\cite{Garey1990computers}. Euclidean TSP has good approximation schemes~\cite{Arora1998PTAS,Mitchell1999PTAS} but these solutions suffer from two potential problems 1) lack of progressive density; 2) cannot support multiple data mules easily. 

\smallskip\noindent\textbf{Random walk.} A practically appealing solution for visiting nodes in a network is by random walk. The downside is that we encounter the coupon collector problem. Initially a random walk visits a new node with high probability. After a random walk has visited a large fraction of nodes, it is highly likely that the next random node encountered has been visited before. Thus it takes a long time to aimlessly walk in the network and hope to find the last few unvisited nodes. Theoretically for a random walk to cover a grid-like network, the number of steps is quadratic in the size of the network~\cite{lovasz96random}. For a random walk of linear number of steps, there are a lot of duplicate visits as well as a large number of nodes that are not visited at all. In the case of multiple random walks, since there is little coordination between the random walks, they may visit the same nodes and duplicate their efforts.

\def\img{\operatorname{Img}}
\def\ker{\operatorname{Ker}}

\section{Theory of constructing space filling curves}

For ease of exposition, we start by summarizing our method in this section. We then describe the theory behind our constructed curve, and end this section with the proof that the curve is dense. For the sake of completeness, we have provided all the theoretical material necessary for understanding our construction.

\subsection{Informal discussion of techniques}
Let us consider the mathematical problem of constructing a dense curve with the desired property of \textit{proportional density} on a two dimensional manifold $S$. We first treat the surface as a one dimensional complex manifold, also called a Riemann surface. This basically means that locally our surface looks like an open set in the complex plane, and the transition maps from one such local ``chart'' to another are holomorphic.

With this point of view, we consider a \textit{holomorphic differential} on our Riemann surface $S$. A holomorphic differential is basically an assignment of a complex-valued holomorphic function on each chart of the surface, that transforms line elements in the correct way; in complex coordinates $z$ and $\bar{z}$, it is a tensor of type $(1,0)$.

Using properties of certain special kinds of holomorphic differentials called Strebel differentials, we partition our surface into pieces, each of which is a flat torus with some holes removed. Each such piece is mapped to a parallelogram with slits (the boundaries of the holes map to the slits). In other words, we view the surface $S$ as a union of parallelograms with slits, with slits being glued together in a certain way. This change of coordinates is mathematically termed a ``branched covering''.

In these coordinates, our curve is just a straight line on the cover. The slope of this line is either irrational, or chosen randomly, depending on the position of the slits and the sides of the parallelograms. Using several important and recent results in \textit{Teichm\"{u}ller theory}, we can prove that this curve is dense. 

Note that although we partition the surface into pieces, we do not cover one piece first and then move on to the next. Instead our curve comes back into each piece infinitely often, increasing the density proportionally to the length.
 
\subsection{Theoretic Background}

\noindent\textbf{Conformal Atlas}
Suppose $(S,\mathbf{g})$ is a surface with a Riemannian metric $\mathbf{g}$. Given any point $p\in S$, there is a neighborhood $U(p)$, one can find the \emph{isothermal coordinates} (i.e.,  local coordinates where the metric is conformal to the Euclidean metric) $(x,y)$ on $U(p)$, such that
\[
    \mathbf{g}= e^{2\lambda(x,y)}(dx^2+dy^2),
\]
where the scalar function $\lambda:U(p)\to \mathbb{R}$ is the conformal factor function. The atlas consisting of isothermal coordinates is called a \emph{conformal atlas}. In the following discussion, we always assume the local parameters are isothermal.

\smallskip\noindent\textbf{De Rham Cohomology}
De Rham cohomology theory is based on the existence of differential forms with certain prescribed properties. Suppose $f:S\to\mathbb{R}$ is a function defined on $S$, then its differential is given by
\[
    df(x,y) = \frac{\partial f(x,y)}{\partial x} dx + \frac{\partial f(x,y)}{\partial y} dy,
\]
Suppose $\omega$ is a differential 1-form on the surface, which has local representation as
\[
    \omega(x,y) = f(x,y)dx + g(x,y) dy,
\]
The \emph{exterior differential operator} $d$ acts on $\omega$,
\[
    d\omega(x,y) = (\frac{\partial g}{\partial x}-\frac{\partial f}{\partial y}) dx\wedge dy,
\]
if $d\omega=0$, then $\omega$ is called a \emph{closed 1-form}. If there exists a function $h: S\to\mathbb{R}$, such that $\omega = dh$, then $\omega$ is called an \emph{exact 1-form}. Exact 1-forms are closed. The first De Rham cohomology group of the surface consists of all non-exact closed 1-forms,
\[
    H^1(S,\mathbf{R}) = \frac{\ker d}{\img d},
\]
where $\ker$, $\img$ represent the kernel and image of the operator $d$.

\smallskip\noindent\textbf{Hodge Decomposition}
The Hodge star operator on differential forms are defined as
\[
    {}^* \omega = {}^*(f(x,y)dx + g(x,y) dy) = (-g(x,y)dx + f(x,y)dy).
\]
A differential 1-form is called a \emph{harmonic 1-form}, if
\[
    d\omega  = 0, d{}^*\omega = 0.
\]
Hodge decomposition theorem claims that each cohomological class has a unique harmonic form. All group consisting of all the harmonic 1-forms is denoted as $H_\Delta^1(S,\mathbb{R})$, which is isomorphic to $H^1(S,\mathbb{R})$.

\smallskip\noindent\textbf{Holomorphic Differentials}
Let $\{(U_\alpha,z_\alpha\}$ be the conformal atlas, where the complex parameter $z_\alpha = x_\alpha + \sqrt{-1} y_\alpha$. Suppose $(U_\beta,z_\beta)$ is another chart, the parameter transition function is $z_\beta(z_\alpha)$ is \emph{holomorphic}, namely, it satisfies the following Cauchy-Riemann equation:
\[
    \left\{
    \begin{array}{lcr}
    \frac{\partial x_\beta}{\partial x_\alpha } &=&  \frac{\partial y_\beta}{\partial y_\alpha }\\
    \frac{\partial x_\beta}{\partial y_\alpha } &=&  -\frac{\partial y_\beta}{\partial x_\alpha }
    \end{array}
    \right.
\]

Let $\Omega$ be a complex differential form with local representation
\[
    \Omega(z_\alpha) = f(z_\alpha) dz_\alpha,
\]
where $f(z_\alpha)$ is holomorphic. A holomorphic 1-form can be decomposed to a pair of conjugate harmonic real differential 1-forms,
\[
    \Omega = \omega + \sqrt{-1}{}^*\omega,
\]
where $\omega$ is harmonic. All the holomorphic differentials form a group $\Omega(S)$, which is isomorphic to $H_\Delta^1(S,\mathbb{R})$.

\smallskip\noindent\textbf{Branched covering} Let $X, Y$ be compact connected topological spaces. A continuous mapping $f:X \rightarrow Y$ is called a branched covering if it is a local homeomorphism everywhere except a finite number of ``branch'' points. In the complex setting, this would mean that a branched covering is, locally at a point $p$, upto composition by biholomorphic maps, of the form $z \rightarrow z^{e_{p}}$, where $e_{p}>1$ for finitely many branch points, and $e_{p}=1$ everywhere else.

\smallskip\noindent\textbf{Trajectory Structure and Strebel differentials} Given a holomorphic 1-form $\Omega$ on a genus $g$ surface, there are $2g-2$ zero points. At each point $p\in S$, the tangent direction $d\gamma \in TM_p$ is called a \emph{horizontal direction}, if $\Omega(d\gamma)$ is real. A curve $\gamma\subset S$ is called a \emph{horizontal trajectory} of $\Omega$, if at each point $p\in \gamma$, $d\gamma$ is along the horizontal direction. The horizontal trajectories through zeros of $\Omega$ are called \emph{critical trajectories}. Similar to holomorphic $1$-forms, one can consider quadratic differentials, which are tensors of type $(2,0)$ in holomorphic coordinates. For quadratic differentials we define a direction to be horizontal if the differential is positive along it, and vertical if it is negative.

If the graph of vertical critical trajectories is compact, the quadratic differential is called Strebel. In the group of quadratic differentials, Strebel differentials are dense~\cite{density_JS}. We will use a holomorphic-$1$ form whose square is Strebel. This will imply that the horizontal trajectories are closed curves.

\subsection{Dense curve construction}

We describe first the branched covering we use to construct our curve, and then prove the density.

\smallskip\noindent\textbf{Branched covering from a Strebel differential} Given a Strebel differential $\Omega$, the critical horizontal trajectories segment the surface to $g$ connected components, denoted as $\{\Gamma_1, \Gamma_2,\cdots, \Gamma_g\}$. Each connected component $\Gamma_k$ is of genus one with boundaries, 
\[
    \partial \Gamma_k = b_k^1 + b_k^2 + \cdots + b_k^{n_k}.
\]
The Strebel differential $\Omega$ induces a flat metric on each $\Gamma_k$, the integration of $\Omega$ on each boundary loop $b_k^i$ maps the boundary loop to a straight line slit. Namely, then integration of $\Omega$ on each $\Gamma_k$ maps $\Gamma_k$ to a flat torus with straight line slits.

The mapping from the surface to the flat tori are diffeomorphic except at the zero points. Locally, the mapping at the zero points is similar to the complex power map $z\mapsto z^2$. Therefore, the zero points are the branch points.

\smallskip\noindent\textbf{The curve we use} Suppose each flat torus is $\mathbb{R}^2/\Gamma_k$, $k=1,2,\cdots,g$. Here $\Gamma_{k}$ represent lattice groups. Then we can find a line $\ell$ on the plane, such that  $\ell$ does not go through any points in the union of lattices $\cup_k \Gamma_k$, since this union is countable. In particular if the lattice points are all rational then a line with irrational slope would do; otherwise we chose a random line. Denote the slope of $\ell$ as $k$. 

On the``welded flat tori", start from one point draw a line $\gamma$ with slope $k$. Then $\gamma$ goes across the handles via the slits; when it hits a slit it moves from one handle to another and continues with the same slope $k$. We take care to chose $k$ in such a way that this line does not pass through the endpoint of any slit. Again, this is easy to maintain since we only have finitely many of these endpoints.

\begin{theorem}
Let $\gamma$ be the curve constructed as above. Then $\gamma$ is dense and does not go through any point more than once.
\end{theorem}

\begin{proof}
Density would imply aperiodicity of $\gamma$, which in turn would imply that $\gamma$ does not pass through any point twice. This is because on each torus, $\gamma$ is a line with slope $k$, and if it visited a point twice it must necessarily become periodic. Thus it suffices to prove that $\gamma$ is dense.

Density of such a curve follows from results in Teichm\"{u}ller theory \cite{chapter_billiards}, and we just sketch the proof. Essentially, as long as the direction $k$ does not contain a ``saddle connection'', which is a trajectory connecting two zeroes of the holomorphic differential, it will be dense. In our case, if the slit coordinates are rational, we choose an irrational $k$; otherwise we choose a random $k$. In both cases, with probability $1$ we will neither hit the lattice points $\Gamma_{k}$ nor the end points of the slits. This guarantees density.
\end{proof}

\section{Discrete algorithm}

Here we first describe for simplicity a centralized algorithm that computes the dense curve. In the next subsection we show how to implement the same algorithm in a distributed fashion.

\subsection{Centralized algorithm}

As mentioned in Section $1$, we assume that the sensors are densely deployed on some underlying surface $S$ such that locally the sensors look like staying on a flat plane. Thus we can apply existing algorithms to first come up with a triangulation of the sensors that approximate the underlying surface $S$~\cite{gao01geometric,funke07sketch}. When sensors have geographical coordinates, we can locally fit a plane at each node and apply the algorithm in~\cite{gao01geometric}. If the sensors do not have geographical coordinates, we can apply coordinate-free algorithm for finding a triangulation, as in~\cite{funke07sketch}.

The surface is approximated by a triangular mesh $M=(V,E,F)$, where $V,E,F$ denotes the vertex, edge and face sets respectively. We use $v_i\in V$ to represent a vertex, $[v_i,v_j]$ an oriented edge from $v_i$ to $v_j$, $[v_i,v_j,v_k]$ an oriented face where $v_i$, $v_j$ and $v_k$ are sorted counter-clock-wisely. We assume the mesh is closed with genus $g$.

The algorithm pipeline is as follows: first, compute the basis of the first homology group $H_1(M,\mathbb{Z})$; second, calculate the dual basis of the first cohomology group $H^1(M,\mathbb{R})$; third, obtain the basis of the harmonic 1-form group; fourth, achieve the basis of the holomorphic 1-form group; finally, by integrating a holomorphic 1-form one gets the required branched covering map. Then we use the curve described in the previous section. We proceed to describe algorithms for each step involved.

\paragraph{Homology Group}

\begin{figure}[h]
\begin{center}
\begin{tabular}{cc}
\includegraphics[width=0.2\textwidth]{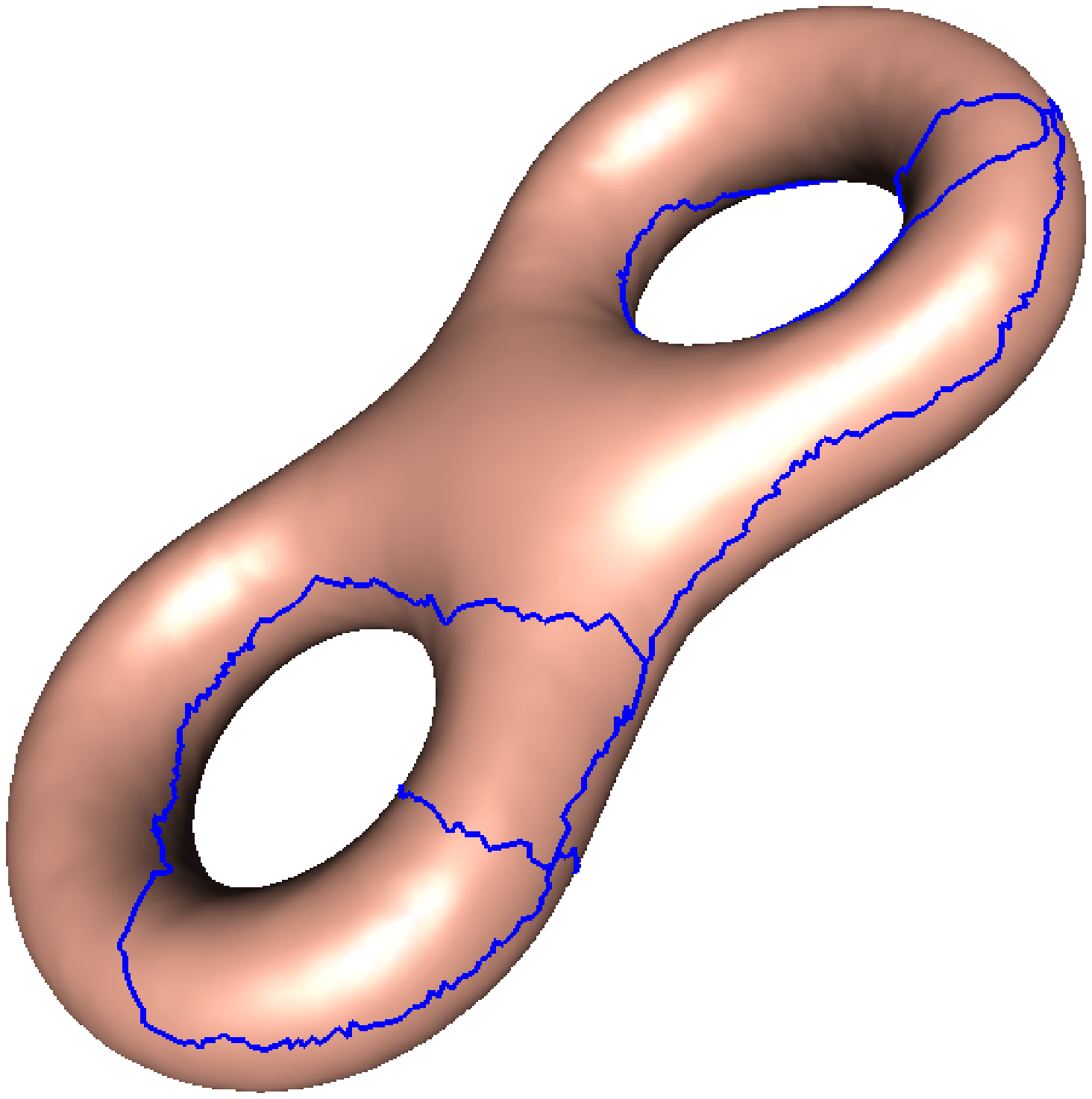}&
\includegraphics[width=0.2\textwidth]{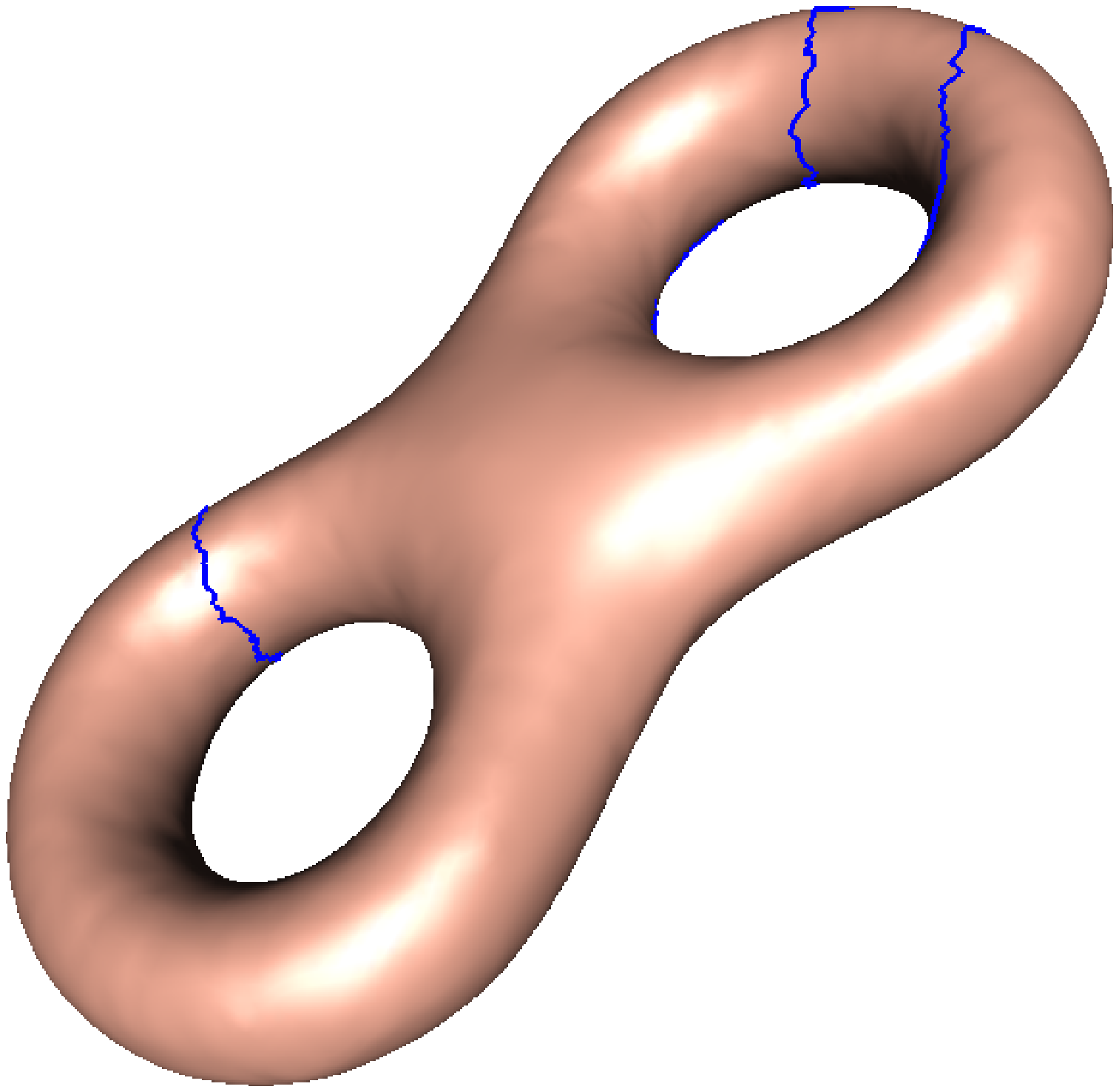}\\
\end{tabular}
\caption{Cut graph. \label{fig:cut_graph}}
\end{center}
\vspace*{-4mm}
\end{figure}

First, the dual mesh $\bar{M}=(\bar{V},\bar{E},\bar{F})$ of the input mesh $M$ is constructed. Each vertex $v_i\in V$, face $f_j\in F$ and edge $e_k\in E$ corresponds to a face $\bar{v}_i \in \bar{F}$, a vertex $\bar{f}_j\in \bar{V}$ and and edge $\bar{e}_k \in \bar{E}$ on the dual mesh respectively.

Second, a spanning tree $\bar{T}$ of the dual mesh $\bar{M}$ is computed. The \emph{cut graph} $G \subset M$ is the union of edges, whose dual edges are not in the spanning tree:
\[
    G = \{ e\in E| \bar{e} \not\in \bar{T}\}.
\]
Intuitively, the mesh $M \setminus G$ with the cut graph removed is a topological disk. See Figure~\ref{fig:cut_graph}. Third, a spanning tree $T$ of the cut graph $G$ is calculated. The complement of $T$ in $G$ is a union of edges:
\[
    G/T = \{e_1,e_2,\cdots, e_{2g}\},
\]
Each edge $e_i$ when included in the spanning tree $T$ (thus $T\cup e_i$) gives rise to a unique loop $\gamma_i \subset T\cup e_i$. These loops
\[
    \{\gamma_1,\gamma_2,\cdots, \gamma_{2g}\}
\]
form the basis of the first homology group $H_1(M,\mathbb{Z})$.

\begin{figure}[h]
\begin{center}
\begin{tabular}{cc}
\includegraphics[width=0.2\textwidth]{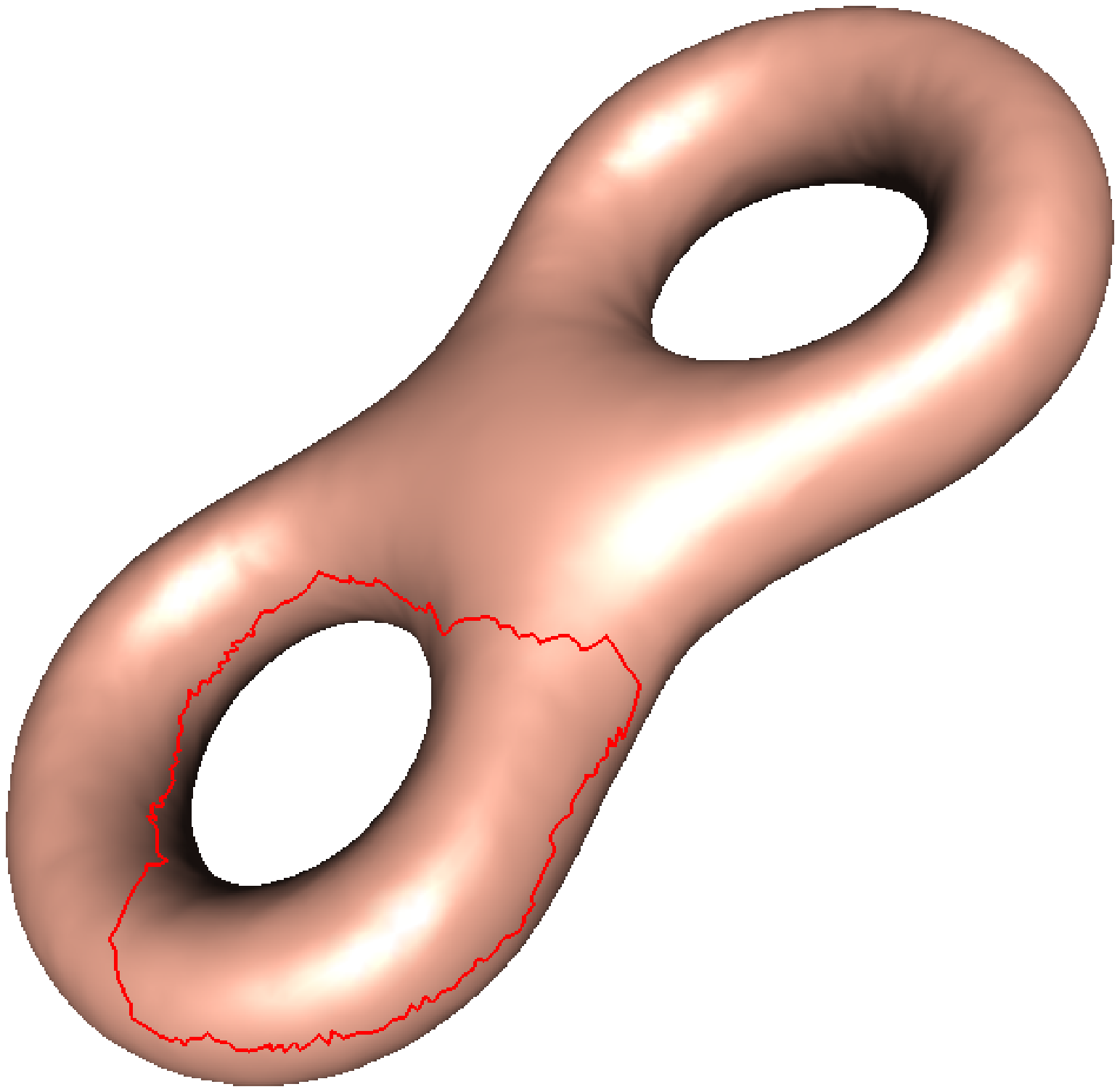}&
\includegraphics[width=0.2\textwidth]{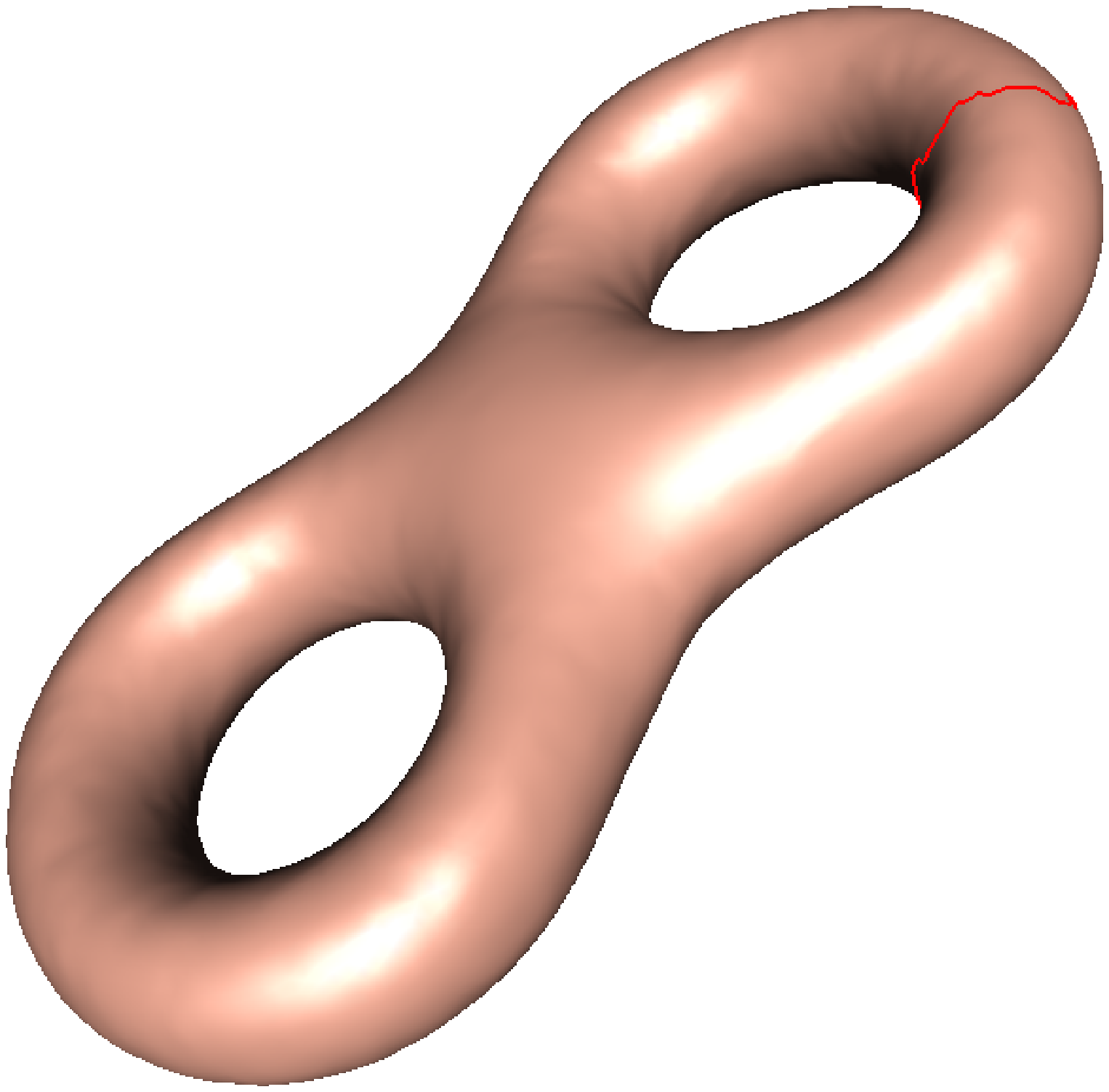}\\
\includegraphics[width=0.2\textwidth]{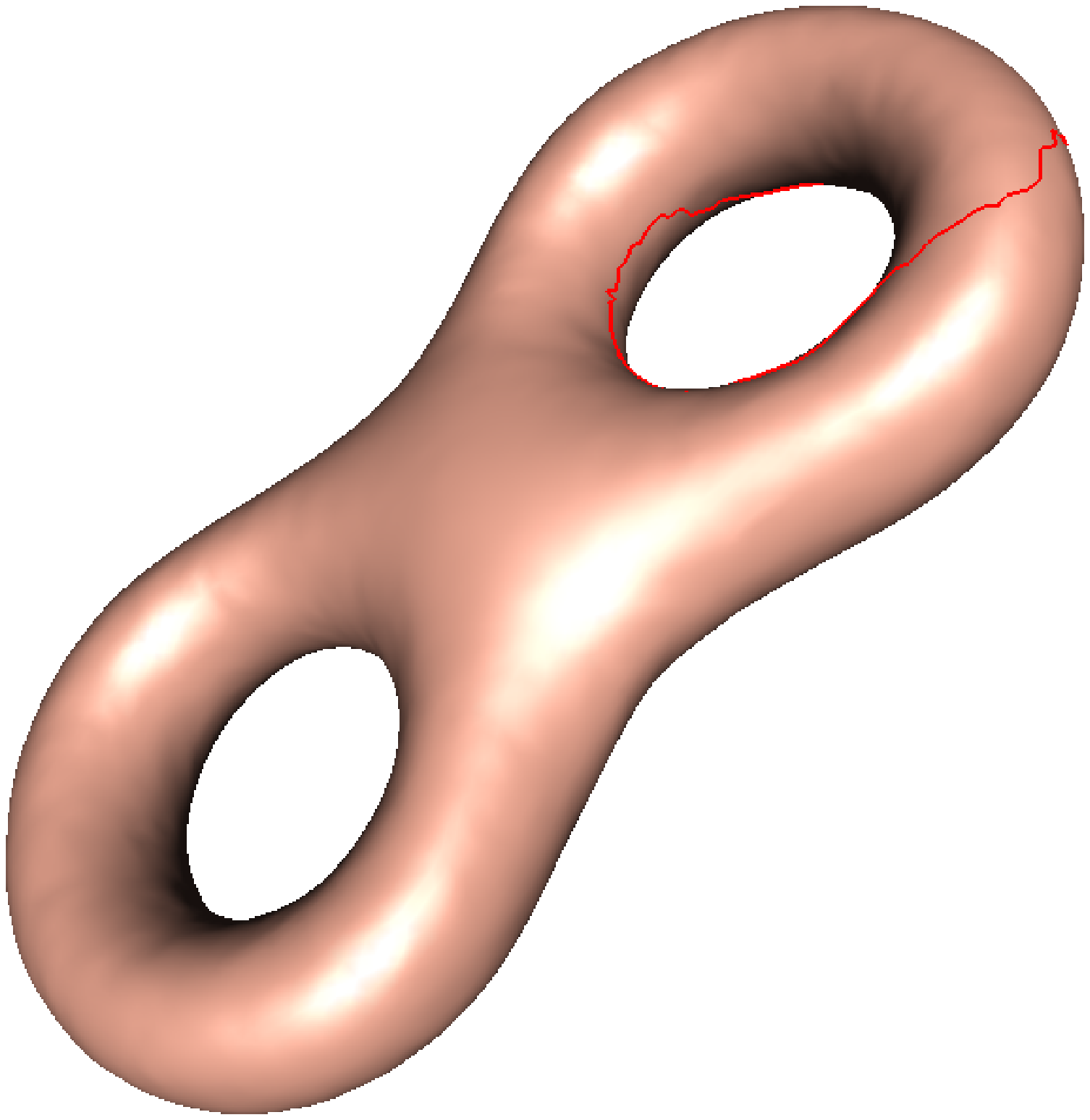}&
\includegraphics[width=0.2\textwidth]{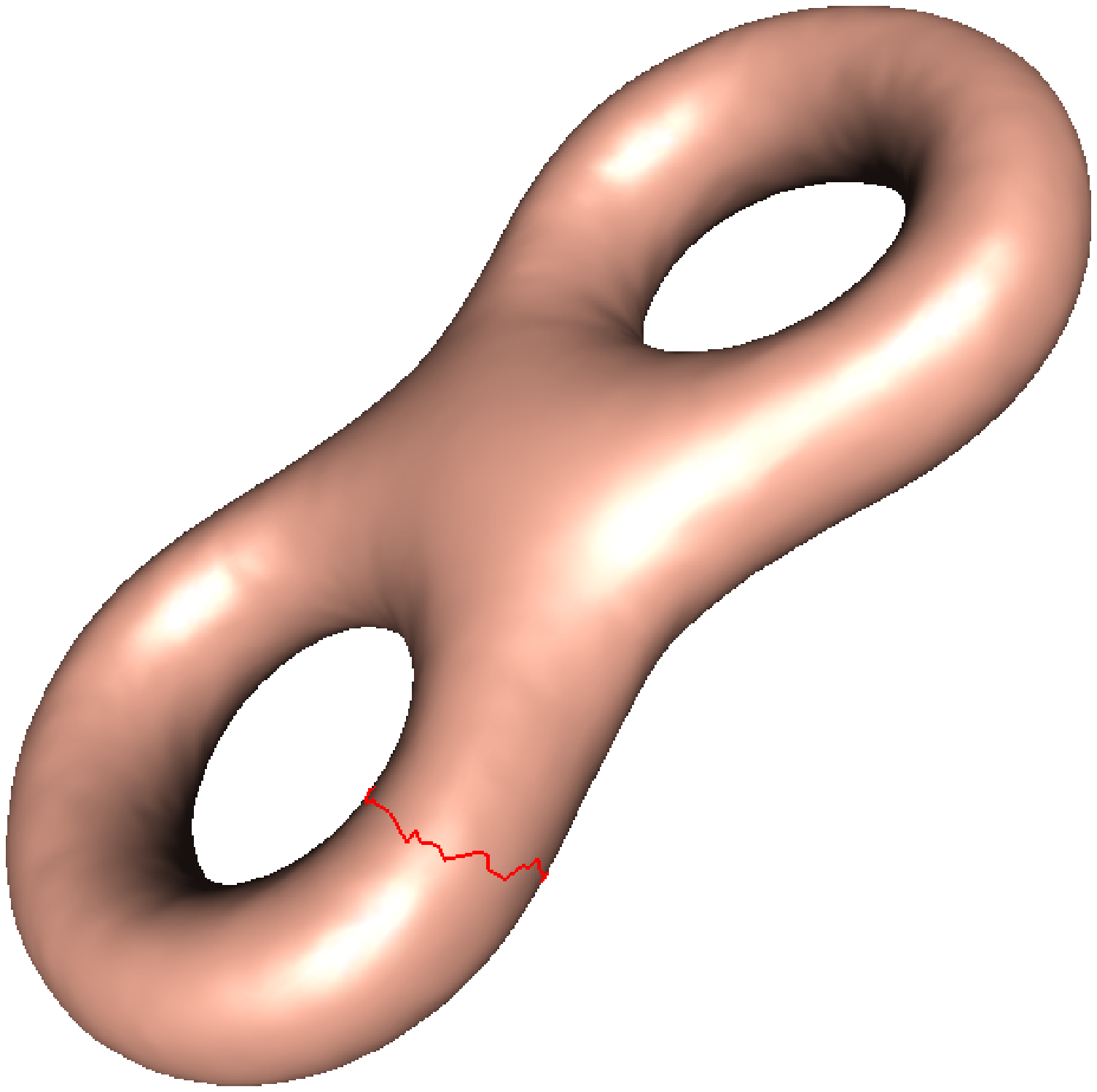}\\
\end{tabular}
\caption{Homology Basis. \label{fig:homology}}
\vspace*{-4mm}
\end{center}
\end{figure}

\paragraph{Cohomology Group}

The differential forms are approximated by discrete forms. A discrete $0$-form  is a function defined on vertices, $f:V\to \mathbb{R}$; a discrete $1$-form is a function defined on the oriented edges, $\omega: E\to \mathbb{R}$, $\omega([v_i,v_j]) = -\omega([v_j,v_i])$; a discrete $2$-form is defined on the oriented faces $\tau: F\to \mathbb{R}$. Discrete exterior differential operator $d$ is dual to the boundary operator $\partial$, for example
\[\begin{tabular}{lll}
    $d\omega([v_i,v_j,v_k])$ &  = &  $ \omega(\partial [v_i,v_j,v_k])$\\
& = & $  \omega([v_i,v_j])+\omega([v_j,v_k])+\omega([v_k,v_i]).$
\end{tabular}
\]

Given the first homology group $H_1(M,\mathbb{Z})$ basis $\{\gamma_1,\cdots,\gamma_{2g}\}$, the dual cohomology group basis can be obtained as follows. For each base loop $\gamma_k$, slice the mesh $M$ to get an open mesh $M_k$. The boundary of $M_k$ has two connected components, denoted them as
\[
    \partial M_k = \gamma_k^+ \cup \gamma_k^-.
\]
Construct a function $f_k :M_k \to \mathbb{R}$,
\[
    f_k(v_i) = \left\{
    \begin{array}{lr}
    1 & v_i \in \gamma_k^+\\
    0 & v_i \in \gamma_k^-\\
   \mbox{rand} & v_i\not\in \partial M_k
    \end{array}
    \right.
\]
Then the $1$-form $df_k$, $df_k([v_i,v_j]) = f_k(v_j)-f_k(v_i)$ is $0$ on all boundary edges. Therefore, one can define $\omega_k$ on the original closed mesh $M$. Suppose $e$ is not on the loop $\gamma_k$, then it has a unique corresponding edge $\tilde{e}$ on $M_k$, define $\omega_k(e) = df_k(\tilde{e})$. If $e\subset \gamma_k$, then let $\omega_k(e)=0$. $\omega_k$ is a closed 1-form, and not exact. These non-exact closed $1$-forms
\[
    \{\omega_1,\omega_2,\cdots, \omega_{2g}\}
\]
form the basis of $H^1(M,\mathbb{R})$. See Figure~\ref{fig:homology}.

\paragraph{Harmonic Differential Group} According to Hodge theory, each cohomological class has a unique harmonic 1-form. Given a cohomology group basis  $\{\omega_1,\omega_2,\cdots, \omega_{2g}\}$, for each closed 1-form $\omega_k$, there is a function $h_k:V\to \mathbb{R}$, such that $\omega_k + dh_k$ is harmonic. By definition, the 1-form is curl free
\[
d(\omega_k + dh_k) = d\omega_k + d^2h_k = 0.
\]
The divergence is
\[
    {}^*d{}^*(\omega_k + dh_k) = 0,
\]
this induces the linear system, for each vertex $v_i\in V$,
\[
    \sum_{[v_i,v_j]\in E} w_{ij} \left( h_k(v_j)-h_k(v_i) + \omega_k([v_i,v_j]) \right)=0,
\]
where $w_{ij}$ is the cotangent edge weight. Suppose edge $[v_i,v_j]$ is shared by two faces $[v_i,v_j,v_k]$ and $[v_j,v_i,v_l]$, then
\[
    w_{ij} = \cot \theta_k^{ij} + \cot \theta_l^{ji},
\]
where $\theta_k^{ij}$ is the corner angle at vertex $v_k$ in triangle $[v_i,v_j,v_k]$. The coefficient matrix of the linear system is positive definite, the solution exists and is unique. The 1-forms
\[
    \{\omega_1 + dh_1 , \omega_2 + dh_2, \cdots, \omega_{2g}+dh_{2g}\}
\]
form the basis of the harmonic 1-form group.

\begin{figure}[h]
\begin{center}
\begin{tabular}{cc}
\includegraphics[width=0.2\textwidth]{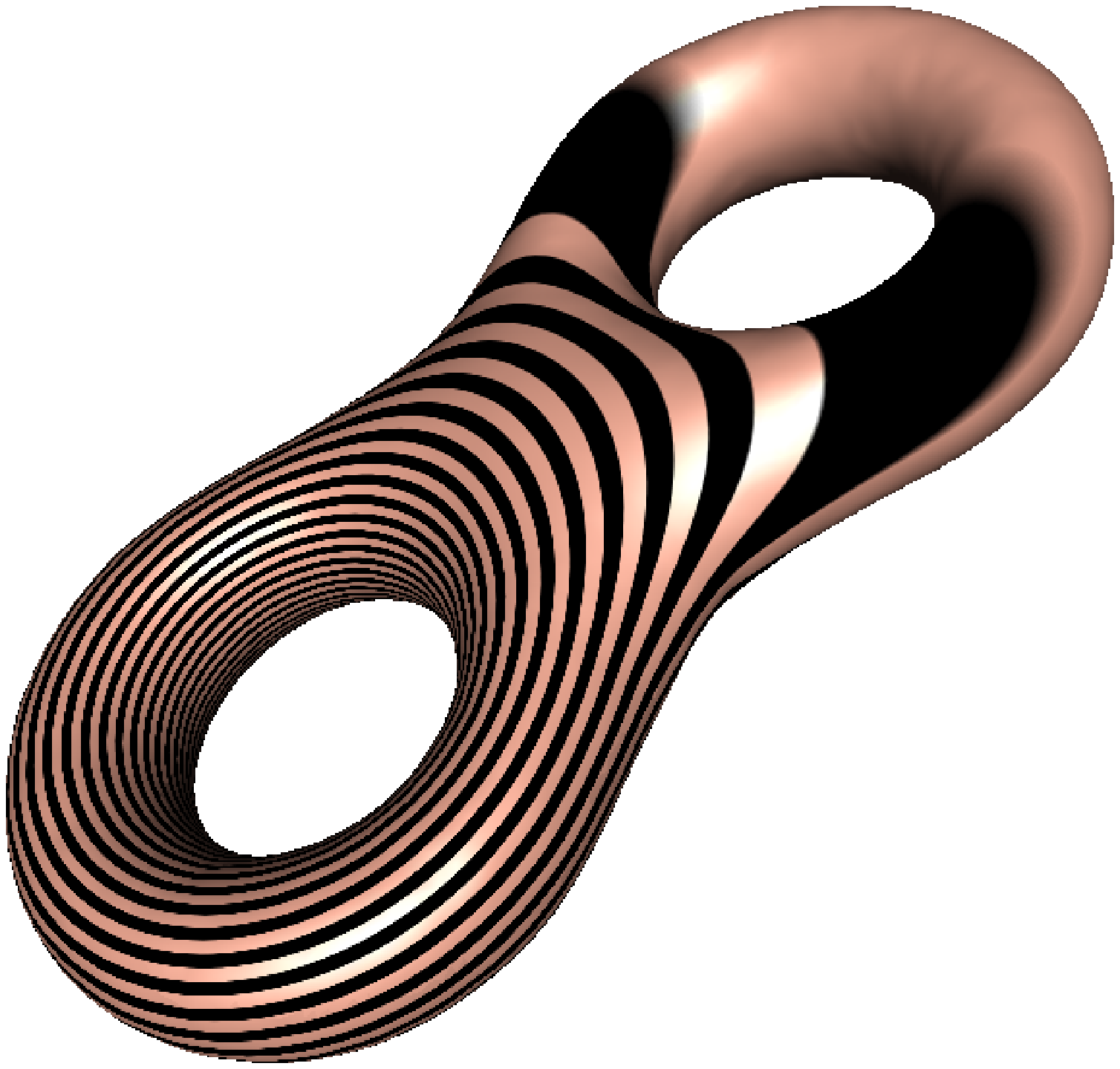}&
\includegraphics[width=0.2\textwidth]{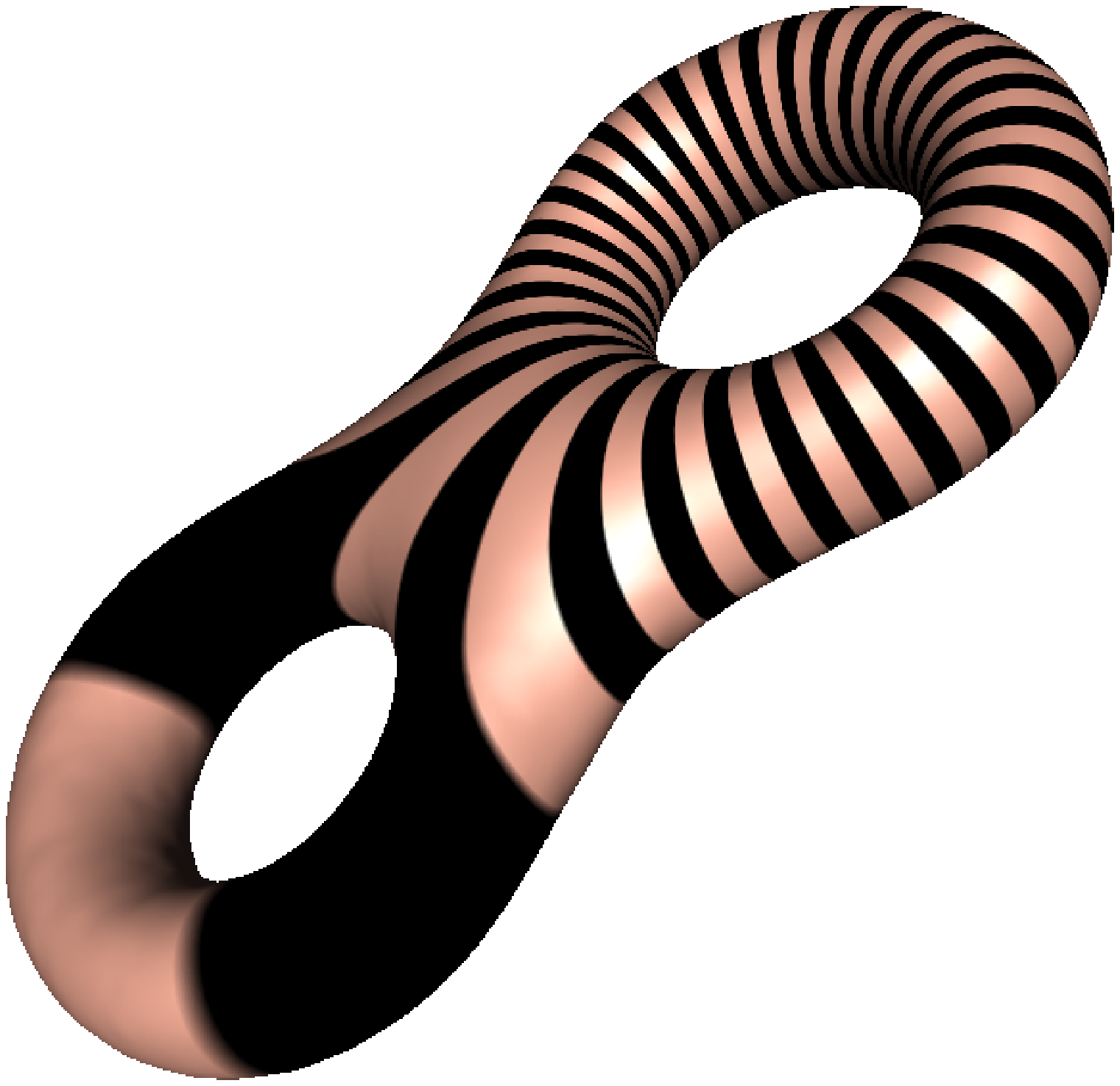}\\
\includegraphics[width=0.2\textwidth]{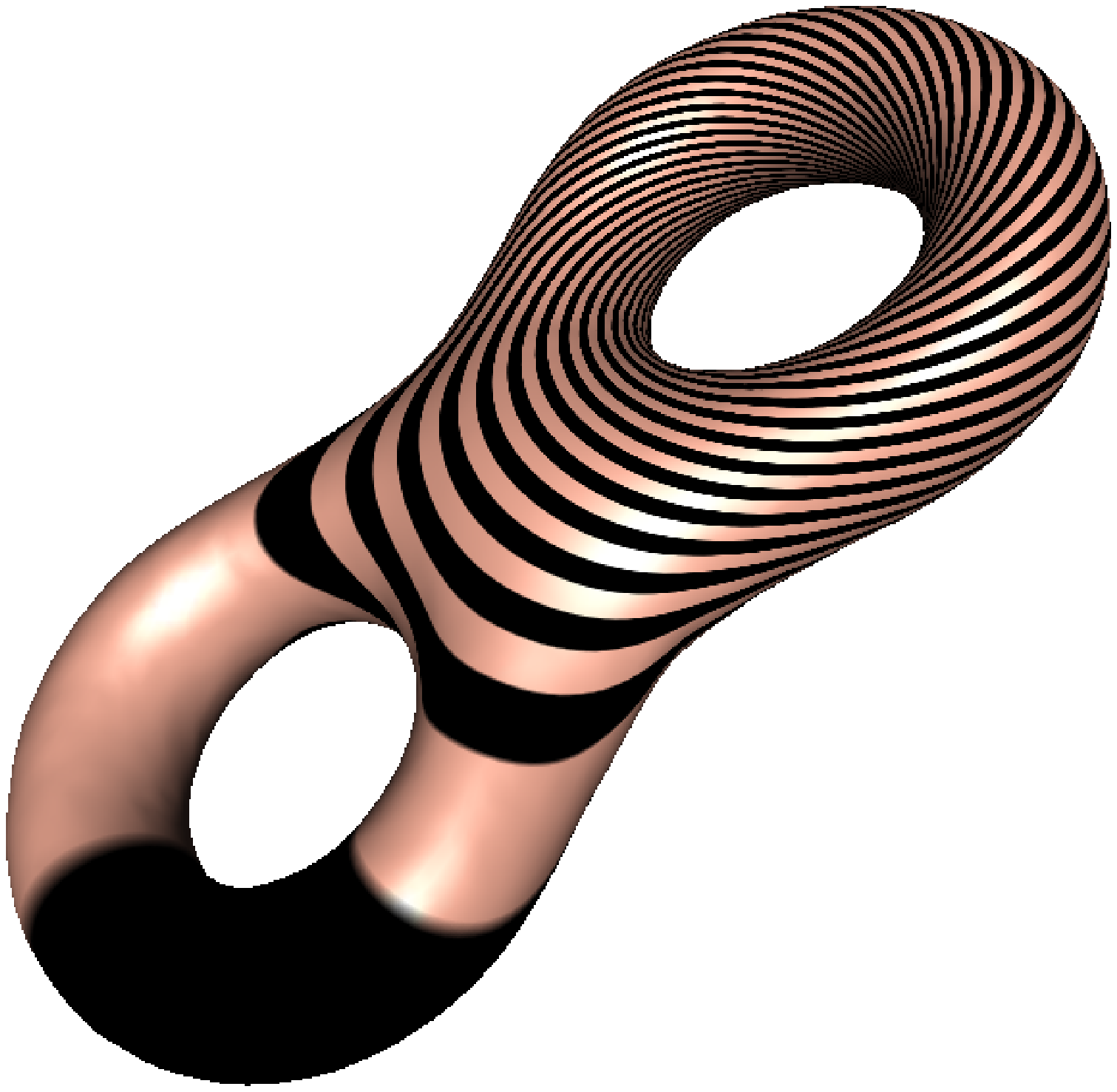}&
\includegraphics[width=0.2\textwidth]{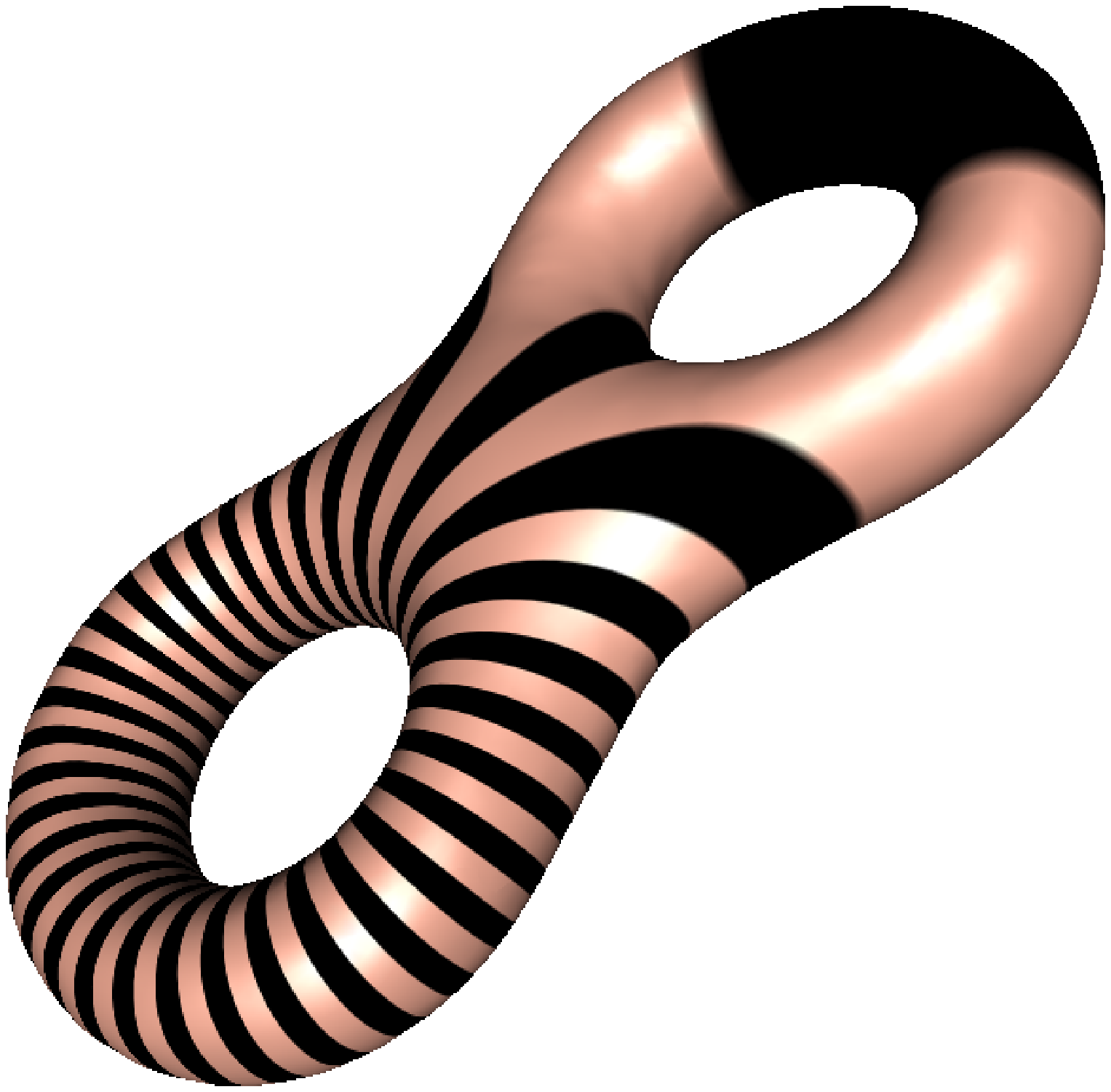}\\
\end{tabular}
\caption{Harmonic 1-form Basis. \label{fig:harmonic}}
\vspace*{-4mm}
\end{center}
\end{figure}

\paragraph{Holomorphic Differential Group} Suppose the harmonic 1-form group basis is given, still denoted as $\{\omega_1,\omega_2,\cdots,\omega_{2g}\}$. Let $\omega$ be a harmonic 1-form, its conjugate 1-form ${}^*\omega$ is harmonic as well, therefore it can be represented as linear combination of $\{\omega_k\}$,
\begin{equation}
    {}^*\omega = \lambda_1 \omega_1 + \lambda_2 \omega_2 + \cdots + \lambda_{2g} \omega_{2g}.
    \label{eqn:hodge_star}
\end{equation}
The coefficients can be obtained by solving the following the linear system
\[
    \int_M {}^*\omega \wedge \omega_ k = \sum_{i=1}^{2g} \lambda_i \int_M \omega_i \wedge \omega_k, k = 1,2,\cdots, 2g.
\]
On one triangle $[v_i,v_j,v_k]$ embed on the plane $\mathbb{R}^2$, the closed 1-form $\omega_k$ can be represented as a constant 1-form $\omega_k = a_k dx + b_k dy$, such that
\[
    \omega_k([v_i,v_j]) = \int_{[v_i,v_j]} a_k dx + b_k dy,
\]
same equations hold for other edges $[v_j,v_k]$ and $[v_k,v_i]$. The wedge product on the face is given by
\[
\omega_i \wedge \omega_j = (a_idx+ b_idy)\wedge (a_jdx+b_jdy) =
\left|
\begin{array}{cc}
a_i&b_i\\
a_j&b_j
\end{array}
\right|dx\wedge dy.
\]
Therefore
\[
    \int_{[v_i,v_j,v_k]} \omega_i \wedge \omega_j = (a_ib_j-a_jb_i)~Area([v_i,v_j,v_k]).
\]
and
\[
    \int_M \omega_i \wedge \omega_j = \sum_{[v_i,v_j,v_k]} \int_{[v_i,v_j,v_k]} \omega_i \wedge \omega_j.
\]
Locally, the Hodge operator is given by
\[
    {}^*\omega_k = {}^*(a_kdx+b_kdy) = a_k dy - b_k dx.
\]
So the coefficients in the linear equation \ref{eqn:hodge_star} can be easily computed. By solving the linear system, the conjugate harmonic 1-form ${}^*\omega$ is obtained.

The harmonic 1-form basis $\{\omega_k\}$, paired with its conjugate harmonic 1-form $\{{}^*\omega\}$ form the holomorphic 1-form basis
\[
    \{\omega_1+\sqrt{-1}{}^*\omega_1, \omega_2+\sqrt{-1}{}^*\omega_2,\cdots, \omega_{2g}+\sqrt{-1}{}^*\omega_{2g}\}
\]

\begin{figure}[h]
\begin{center}
\begin{tabular}{cc}
\includegraphics[width=0.2\textwidth]{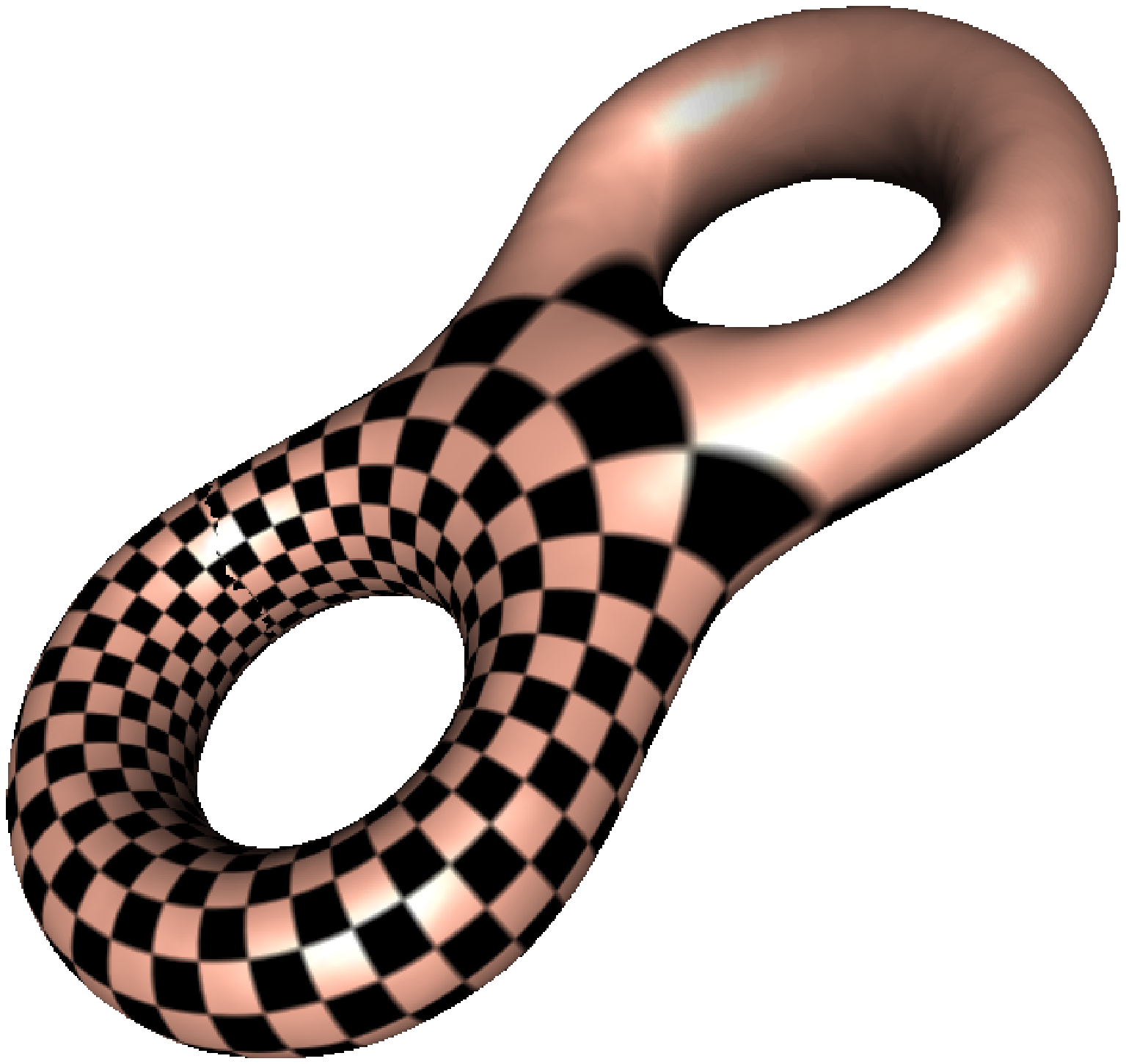}&
\includegraphics[width=0.2\textwidth]{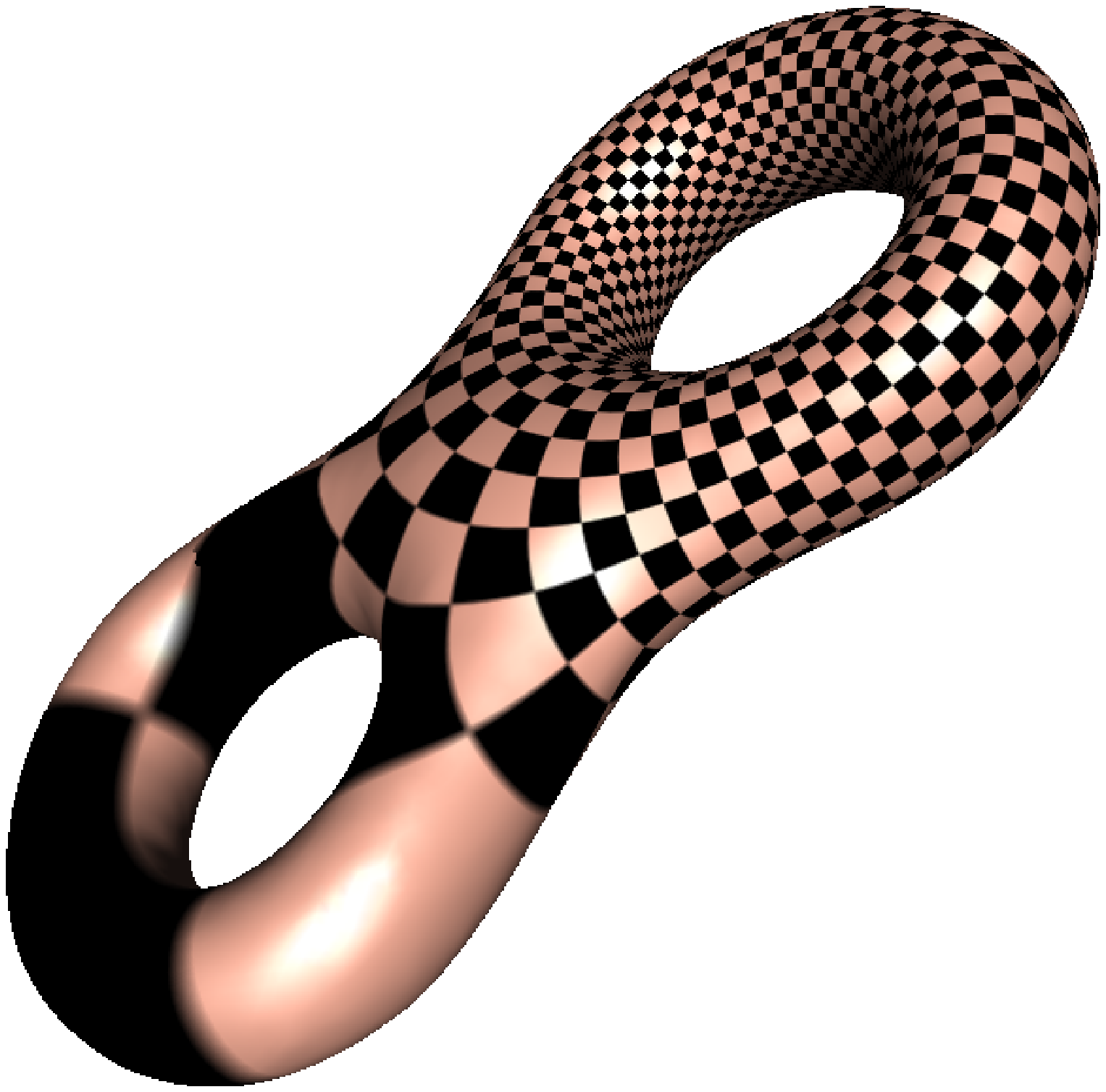}\\
\includegraphics[width=0.2\textwidth]{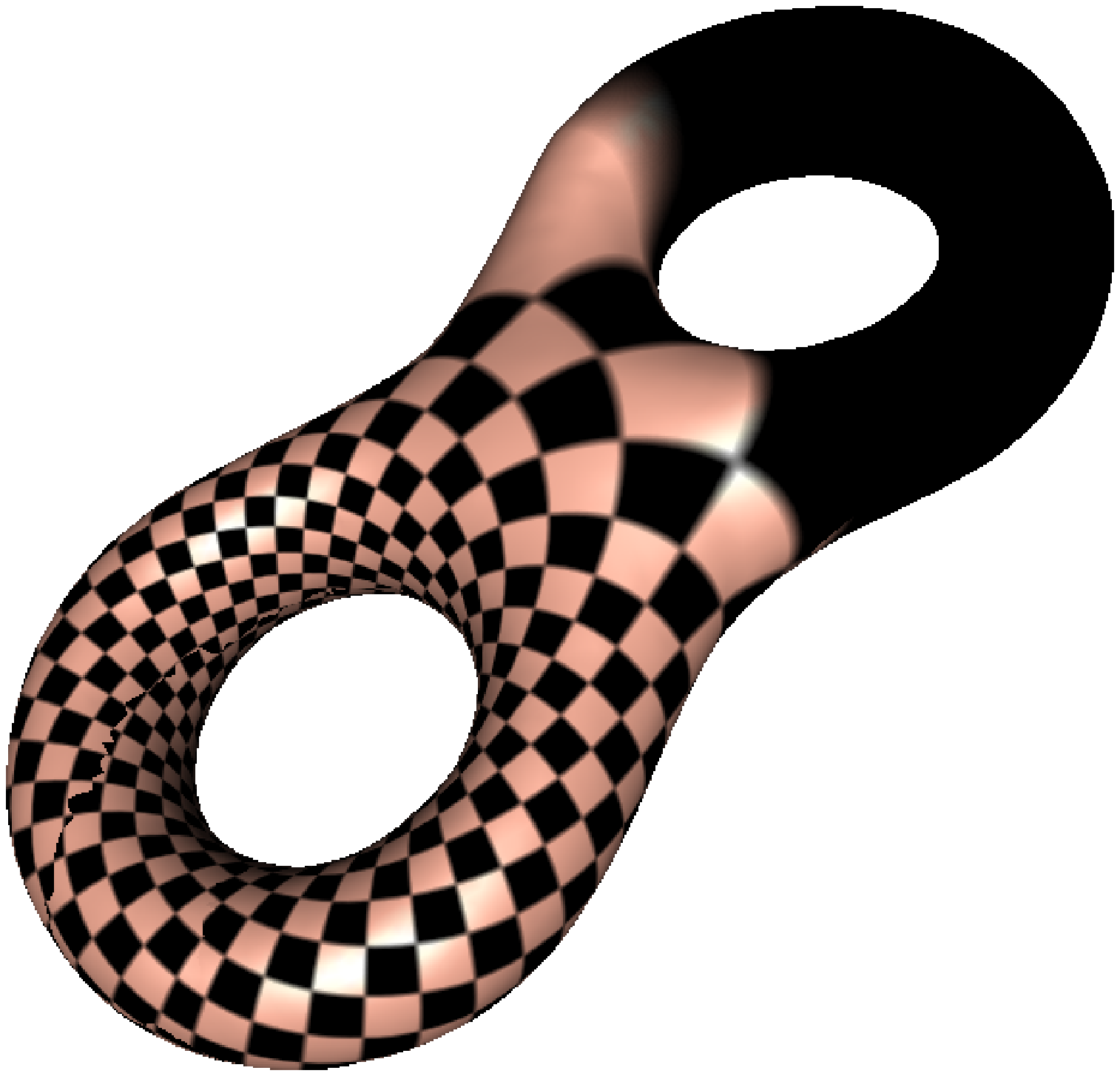}&
\includegraphics[width=0.2\textwidth]{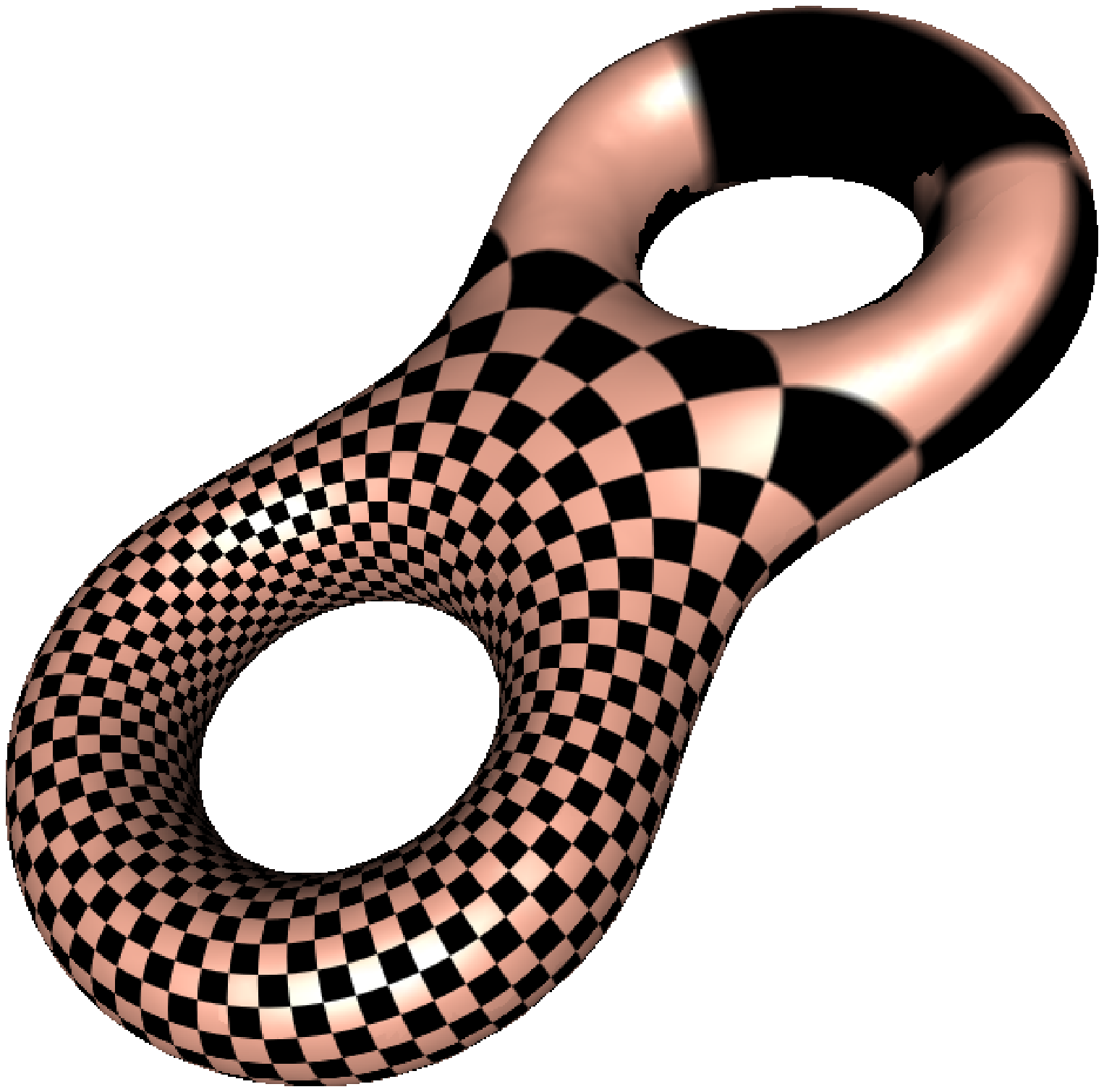}\\
\end{tabular}
\caption{Holomorphic 1-form Basis. \label{fig:harmonic}}
\end{center}
\vspace*{-4mm}
\end{figure}

\paragraph{Branch Covering Map} Compute the cut graph $G$ of the mesh, slice the mesh along the cut group to obtain a \emph{fundamental domain} $M/G$. Choose one holomorphic 1-form $\omega + \sqrt{-1}{}^*\omega$ and integrate the holomorphic 1-form on the fundamental domain to get the branch covering map. Fix a vertex $v_0 \in M/G$ as the base vertex, for any vertex $v_i \in M/G$,
\[
    \varphi(v_i) = \int_{v_0}^{v_i} \omega + \sqrt{-1}{}^*\omega,
\]
the integration path $\gamma \subset M/G$ can be chosen arbitrarily, which consists a sequence of consecutive oriented edges, connecting $v_0$ to $v_i$, denoted as
\[
    \gamma = e_0 + e_1 + \cdots e_k,
\]
such that target vertex of $e_i$ equals to the source vertex of $e_{i+1}$, the source of $e_0$ is $v_0$, the target of $e_k$ is $v_i$.
\[
    \int_\gamma \omega = \sum_{i=0}^k \omega(e_i).
\]
The branching points of $\varphi$ are the zero points of the holomorphic 1-form. The slits are the horizontal trajectories connecting the zeros of the holomorphic 1-form.

\begin{figure}[h]
\begin{center}
\begin{tabular}{cc}
\includegraphics[width=0.2\textwidth]{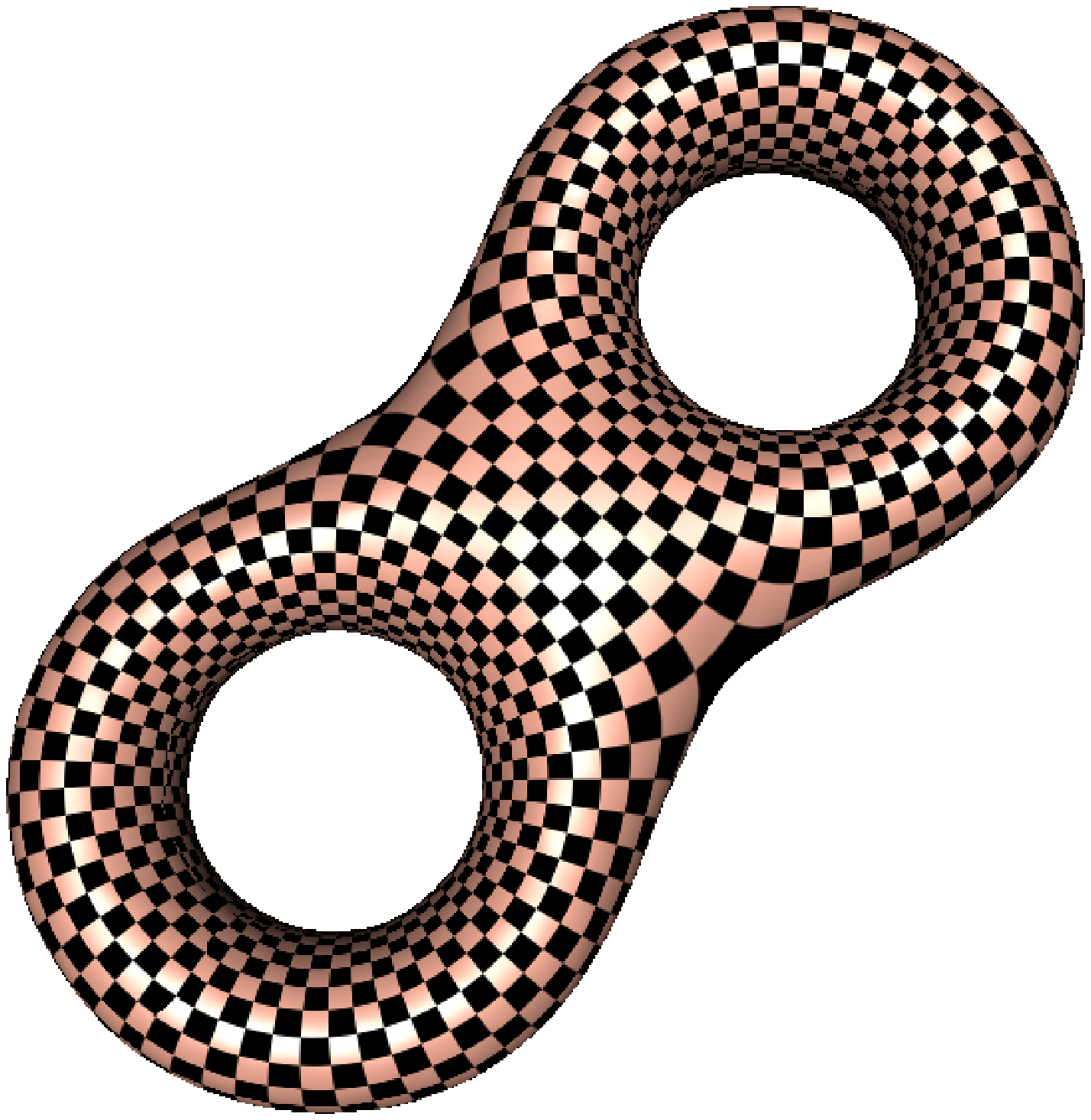}&
\includegraphics[width=0.2\textwidth]{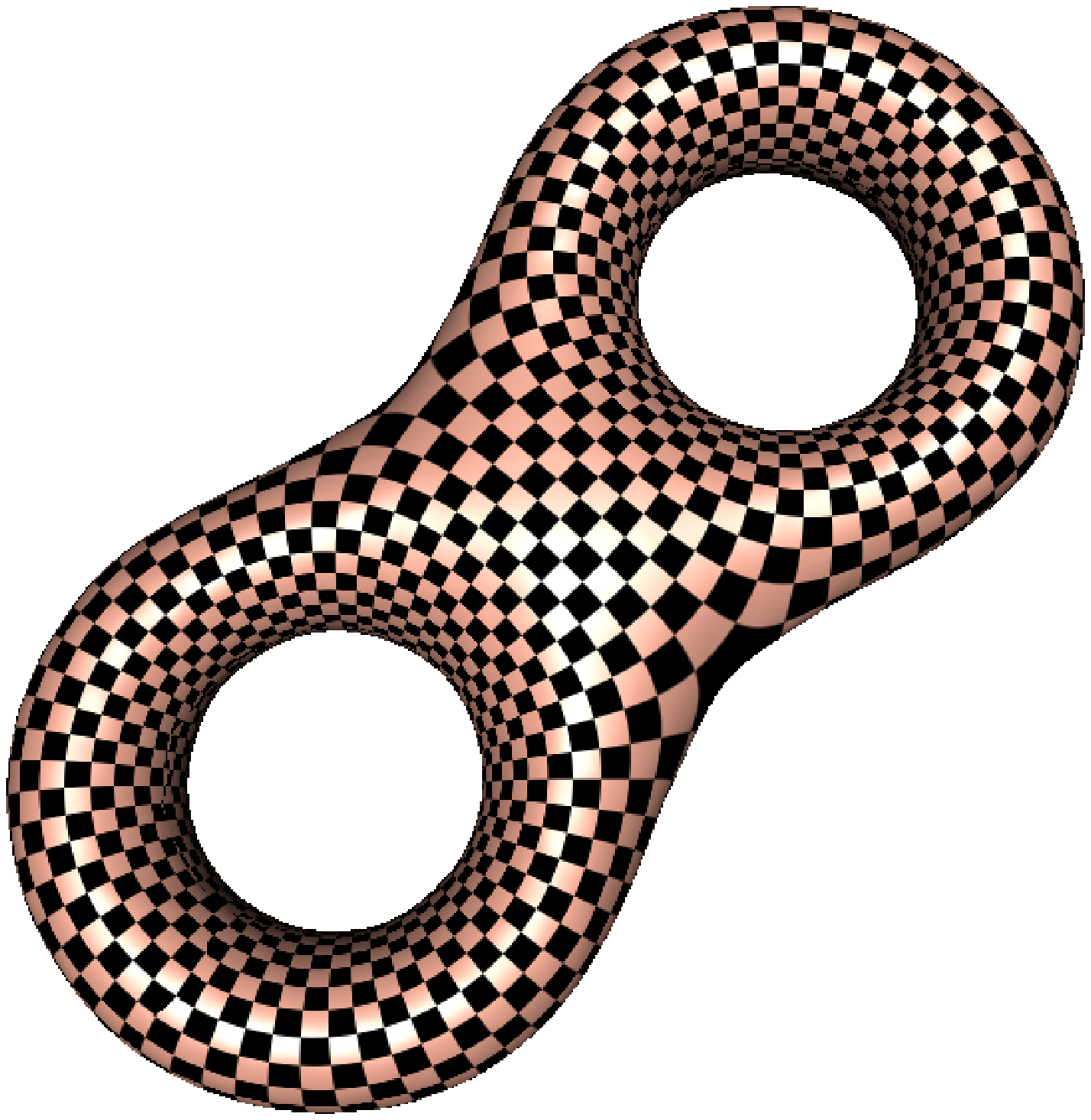}\\
(a) front view & (b) back view \\
\includegraphics[width=0.2\textwidth]{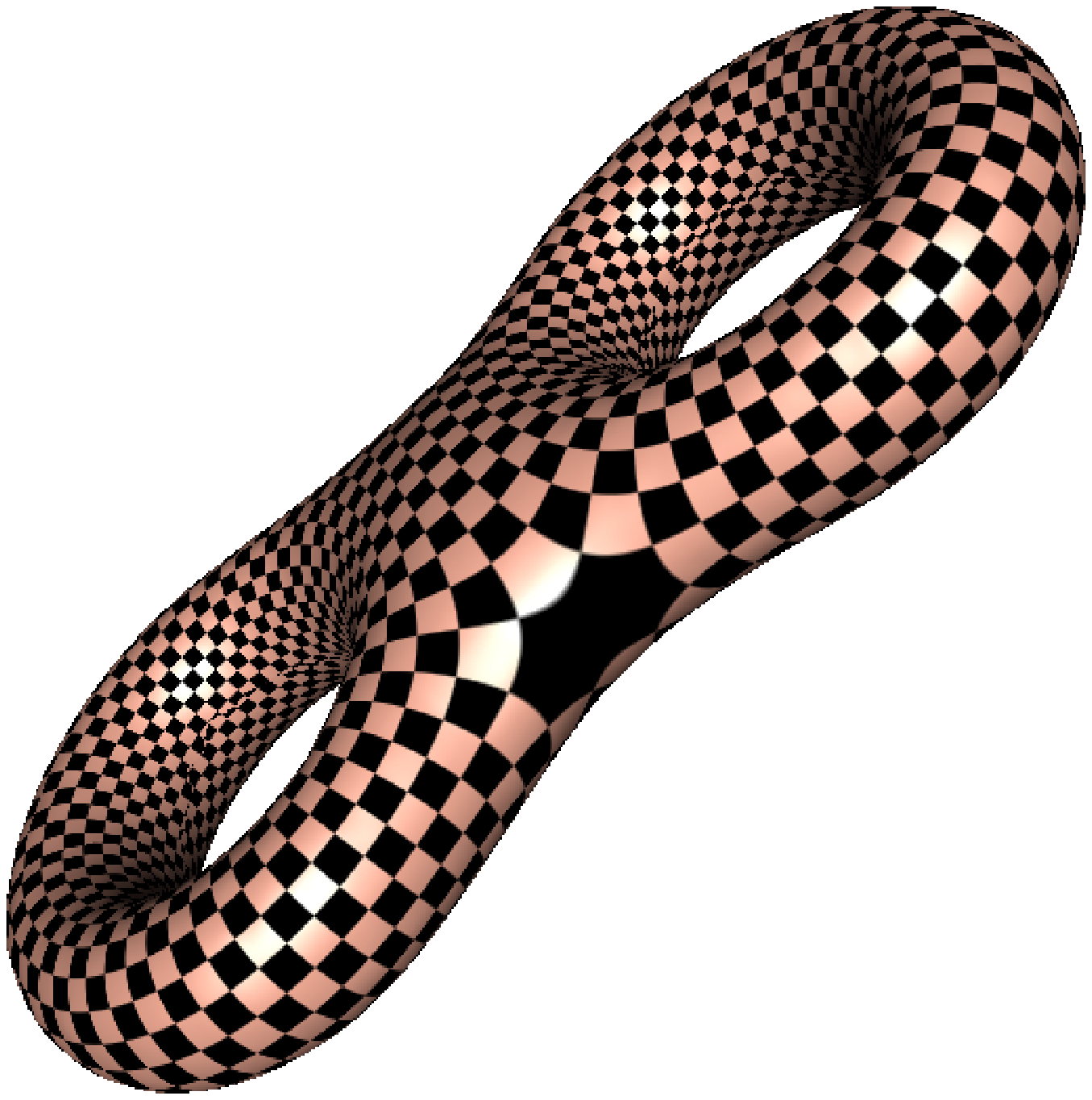}&
\includegraphics[width=0.2\textwidth]{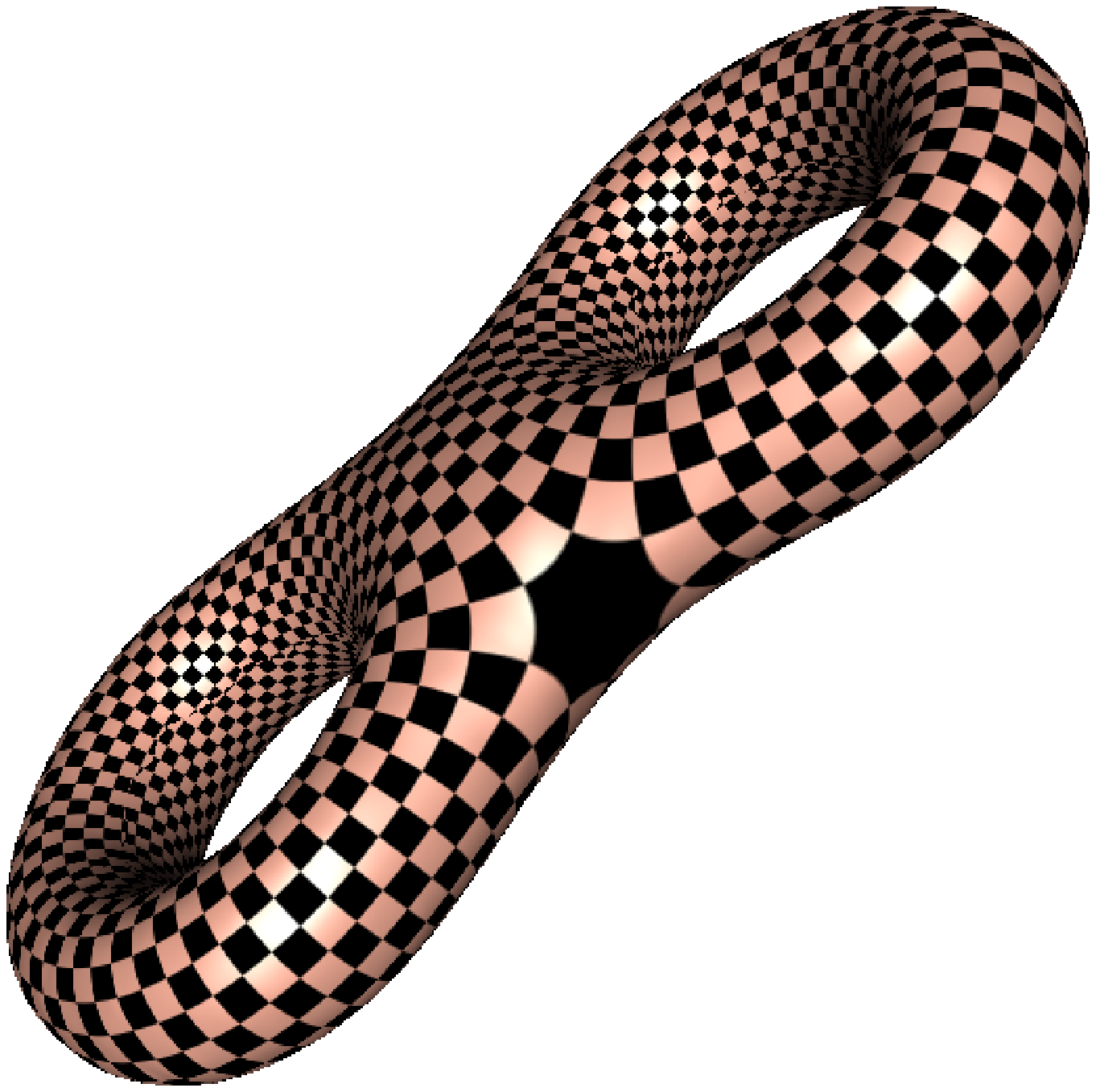}\\
(c) first zero point & (d) second zero point\\
\end{tabular}
\caption{Holomorphic 1-form and zero points. \label{fig:holomorphic_form_zero}}
\end{center}
\vspace*{-4mm}
\end{figure}

\begin{figure}[h]
\begin{center}
\begin{tabular}{cc}
\includegraphics[width=0.2\textwidth]{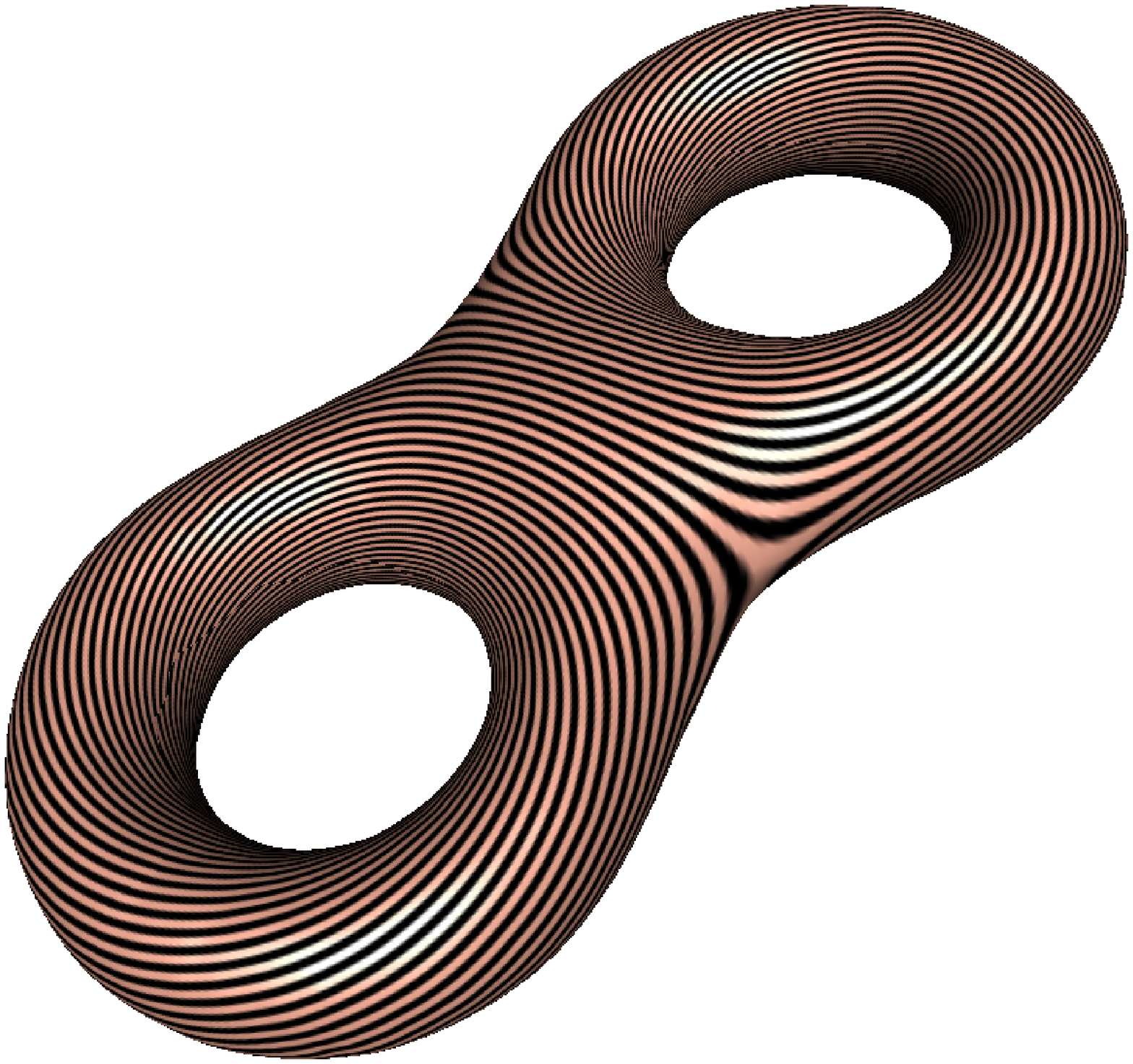}&
\includegraphics[width=0.2\textwidth]{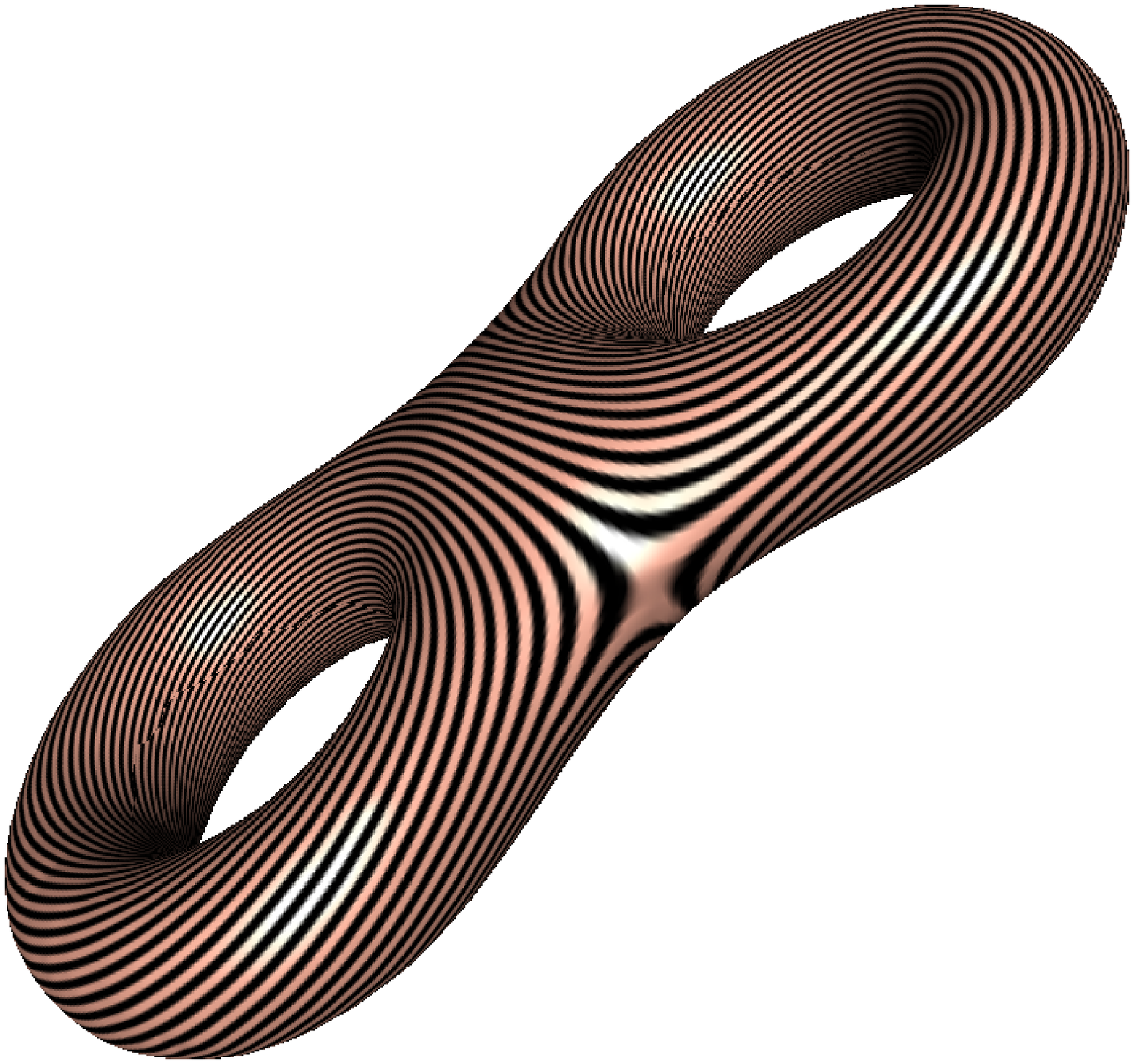}\\
(a) dense curve & (b) zero point \\
\end{tabular}
\caption{Dense curve on the surface. \label{fig:foliation}}
\end{center}
\vspace*{-4mm}
\end{figure}

As shown in Fig.\ref{fig:map}, there are $2g-2$ zero points of the holomorphic 1-form. The horizontal trajectories through the zeros segment the surface into handles as shown in Frame (a). Each handle is conformally mapped onto a flat torus with a slit, the end points of the slit are the zero points, as shown in Frame (b). The flat tori are glued together through slits, the top (bottom) edge of the slit on one torus is glued to the bottom (top) edge of the slit on the other torus, as shown in Frame (c). In the neighborhood of each zero point, the mapping is a branch covering similar to $z\mapsto z^2$, as illustrated in Frame (d).

\begin{figure}[h]
\begin{center}
\begin{tabular}{cc}
\includegraphics[width=0.2\textwidth]{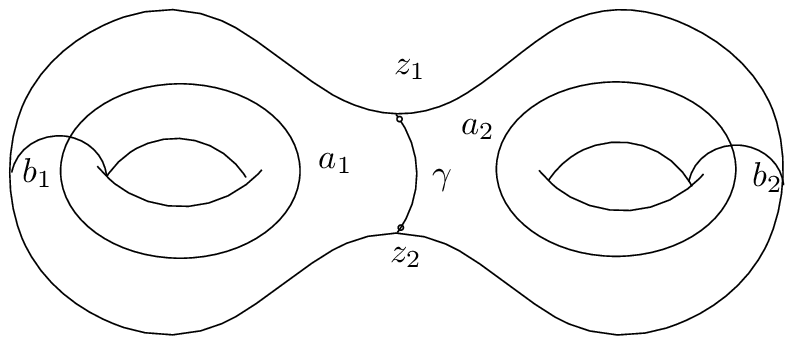}&
\includegraphics[width=0.2\textwidth]{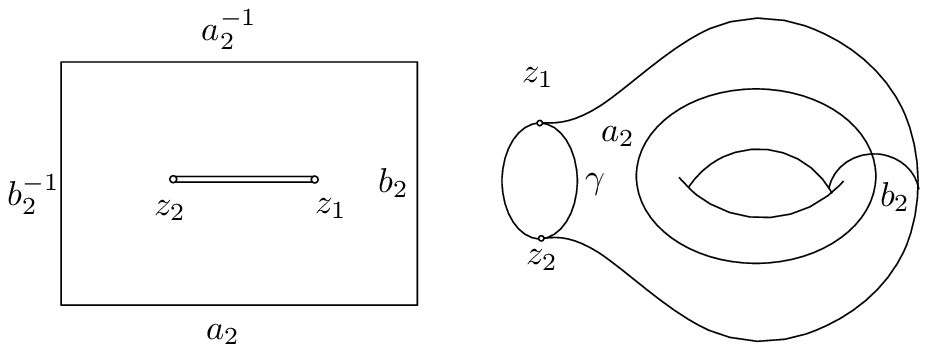}\\
(a) surface & (b) one handle \\
\includegraphics[width=0.2\textwidth]{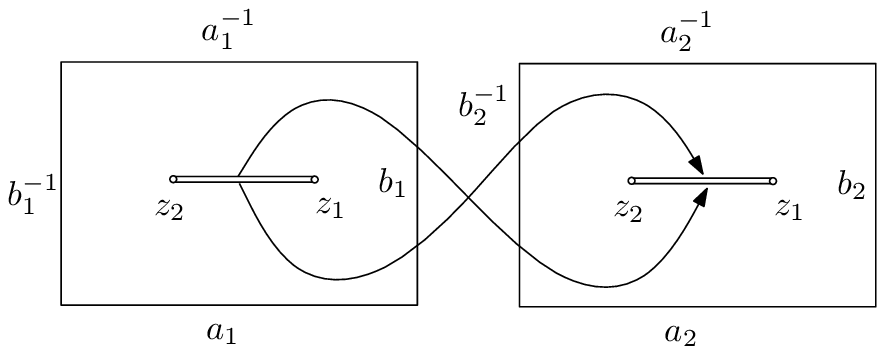}&
\includegraphics[width=0.2\textwidth]{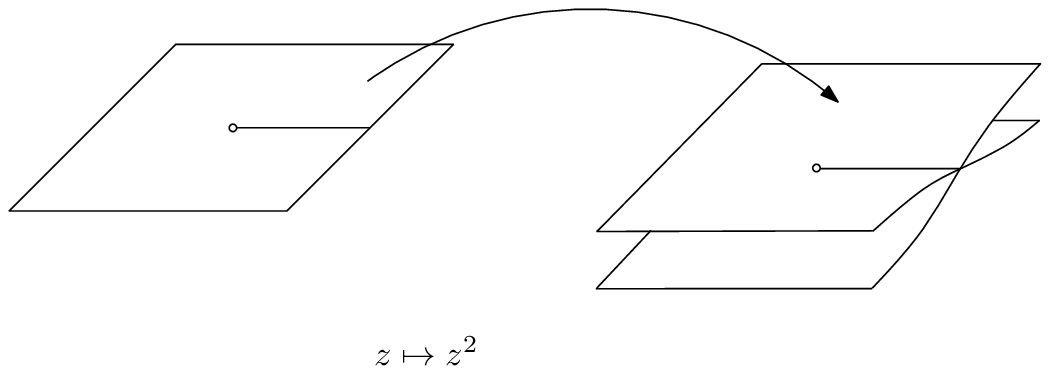}\\
(c) glue pattern & (d) branch point \\
\end{tabular}
\caption{Branch covering map. \label{fig:map}}
\vspace*{-4mm}
\end{center}
\end{figure}

\subsection{Distributed Algorithm}

The centralized algorithms described in the last section can be implemented in the distributed setting as follows

\paragraph{Cut Graph} One node is set as the ``seed node''. This node then sends a message to its neighbors. The message is propagated to the entire network by a single round of flooding. The cut graph is the locus where the wave-front meets, or the nodes which get the message from different sources simultaneously. The details of detecting cut locus can be found in~\cite{wang06boundary}.

\paragraph{Homology Basis} The algorithm boils down to compute a spanning tree of the cut graph. This can be done by breadth-first search by a single round of flooding. Let $G$ be the cut graph, $T$ be the spanning tree, $e$ an edge of $G$ but not on $T$. Then trace back two paths from the end nodes of $e$ in the tree to the root. When the two paths intersect, a loop is obtained.

\paragraph{Cohomology Basis} Let $\gamma$ be a homology basis, make a copy for each node on $\gamma$ and slice the network along $\gamma$. Call the cutting curves that now become two boundary loops as $\gamma^+$ and $\gamma^-$. Construct a function $f$ such that on each node on $\gamma^+$, $f$ equals $1$ and on each node on $\gamma^-$, $f$ equals $0$. $f$ is random on other nodes. Then the gradient of $f$ is a closed 1-form on the original network, dual to $\gamma$.

\paragraph{Harmonic $1$-form Basis} Given a closed 1-form $\omega$, this step is to find a function $f$, such that $ \omega + df$ is harmonic. This can be achieved by a distributed heat diffusion method. At the beginning, $f$ equals to $0$ everywhere. Then for each node $v_i$, $f(v_i)$ is updated by
\[
    f(v_i) \leftarrow \frac{w_{ij}}{\sum_k w_{ik}}\left(f(v_j) + \omega([v_i,v_j])\right).
\]
This diffusion process will converge to the unique harmonic 1-form cohomological to $\omega$.

\paragraph{Hodge Star} Locally, a harmonic 1-form can be represented as a vector field, the Hodge star operator ``rotates'' it about the normal vector of the surface by a right angle. Formally, ${}^* dx = dy$, ${}^*dy = -dx$. 

\paragraph{Branched Covering Map} Given a harmonic 1-form $\omega$, and its conjugate harmonic 1-form ${}^*\omega$, one can integrate them by flooding. Choose one root node $v_0$, set $\varphi(v_0)=(0,0)$. Suppose $\varphi(v_i)$ has been computed, $v_j$ is in its neighbors
\[
    \varphi(v_j) = \varphi(v_i) + (\omega([v_i,v_j]), {}^*\omega([v_i,v_j])).
\] 
Then $\varphi$ gives the branched covering map.
\section{Simulations}

In this section, we simulate the discrete path that traverses and linearizes the sensor network densely distributed on a 2-hole torus in 3D. The communication graph follows the unit disk graph model. The network has a total of $56591$ nodes with average degree $5$, as shown in Figure~\ref{fig:networkFig}(a). We will compare the path generated based on our algorithm with Eulerian cycle and random walk. 

\subsection{Discrete path generation}

\textbf{Dense curve.} We map the genus 2 surface network $S^\prime$ to the canonical, symmetric, torus $S_1$ and $S_2$ with slits, as shown in Figure~\ref{fig:networkFig}(b). Specifically, the nodes on torus $S^\prime_1$ are mapped to $S_1$, and the nodes on torus $S^\prime_2$ are mapped to $S_2$. The zero points on $S^\prime$ form the slits. We take a straight line $\ell$ with slope $e$. When $\ell$ hits the slit, it crosses to the other torus, e.g., from $S_1$ to $S_2$ or \textit{vice versa}. Mapping $\ell$ back to the original network $S^\prime$, we get the dense curve ${\ell}^\prime$ that spirals around the surface. To find a discrete path on the sensor nodes, we expand along the curve ${\ell}^\prime$ to a belt $B$ of width $\delta$. The path starts from an arbitrary node $s$ inside $B$ and always goes to its closet neighbor within the belt. 
This way, we generate a discrete path $P$ visiting the sensor nodes by following the trajectory of ${\ell}^\prime$ in a greedy manner. The path is shown in Figure~\ref{fig:networkFig}(a). 

\smallskip\noindent\textbf{Eulerian cycle.} We build a spanning tree of the network and double all the edges to generate an Eulerian cycle. By traversing along this cycle, we can get a discrete path through the nodes.

\smallskip\noindent\textbf{Random walk.} In the random walk scheme, the next hop of the path is obtained by uniformly randomly choosing a neighbor of the current visiting node.

\subsection{Performance of different discrete paths}

We compare the following two metrics and examine the trend of the metrics when the path length (number of hops) increases.
(1) \textbf{Coverage rate:} the percentage of nodes already visited. (2) \textbf{Average shortest distance:} For each node not yet visited, we calculate the smallest number of hops it could be reached from the current path. The average shortest distance is the average of this value for all the unvisited nodes.

Figture~\ref{fig:coverageFig} shows the network coverage rate as the path gets longer. 
The coverage rates of our method and Eulerian cycle are much better than random walk. Both our method and Eulerian cycle could cover all the nodes when the path length is roughly twice the network size. The coverage rate grow much faster for our method compared to Eulerian cycle except at the very end. 

Specifically, our method could cover 90\% of the nodes with the path length is at about around 1.25 times the network size.

Figure~\ref{fig:distanceFig} shows the average distance between visited nodes and unvisited nodes. 
We can see that the average shortest distances of both Eulerian circle and random walk are much worse than our method. After covering around $5000$ nodes, the average distance to the unvisited nodes under our method is about $3.5$, while the distances are $16.6$ and $22.4$ by Eulerian cycle and random walk respectively. Our method could reach all the nodes in $2$ hops on average when only a quarter of the nodes are visited. We can learn from the results that the path generated by our method could uniformly cover the network.

\begin{figure}[ht]
\vspace{-0.3cm}
\centering
    \begin{tabular}{cc}
    \includegraphics[width=0.20\textwidth]{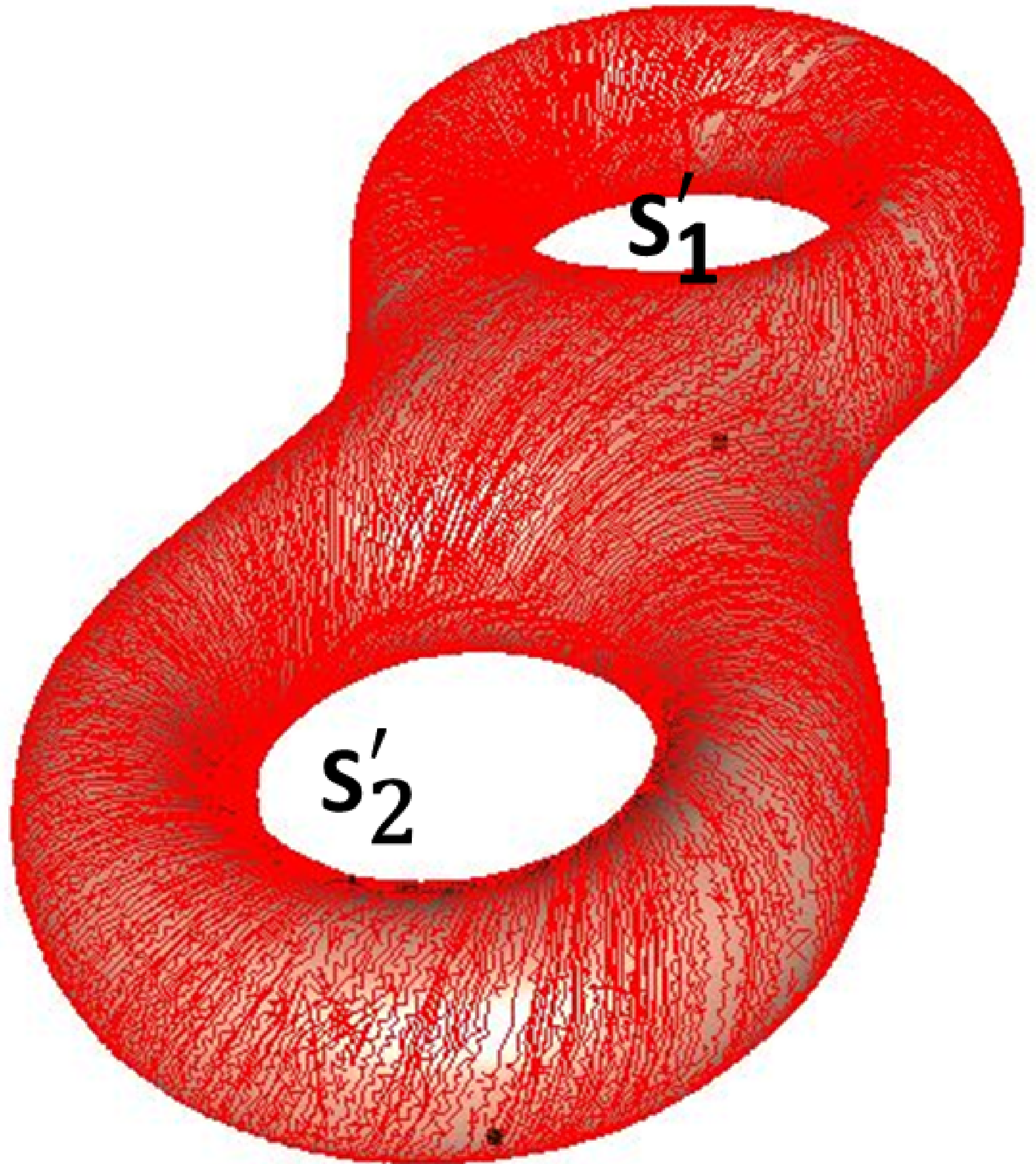}&
    \includegraphics[width=0.22\textwidth]{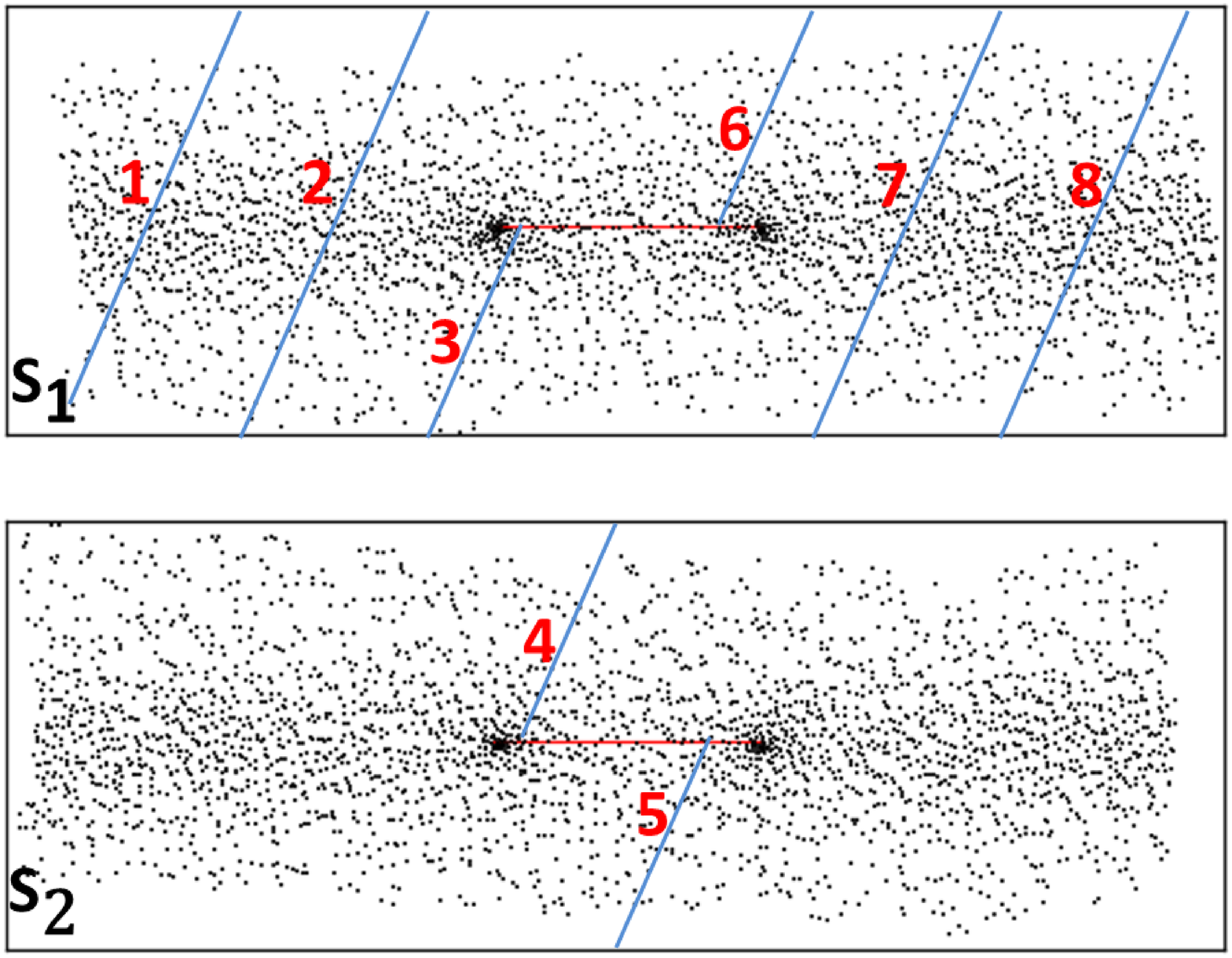}\\
    \footnotesize (a) Dense curve on sensor network& \footnotesize (b) Two torus with slits. \\
    \end{tabular}
  \vspace*{-2mm}
\caption{The sensor network and its related torus with slits. The lines on (a) form the path of dense curve on the network. The line number on (b) explains how to get the dense curve on the two torus.}
 \label{fig:networkFig}
\vspace{-0.3cm}
\end{figure}

\begin{figure}[h]
\begin{center}
\begin{tabular}{cc}
\includegraphics[width=0.45\textwidth]{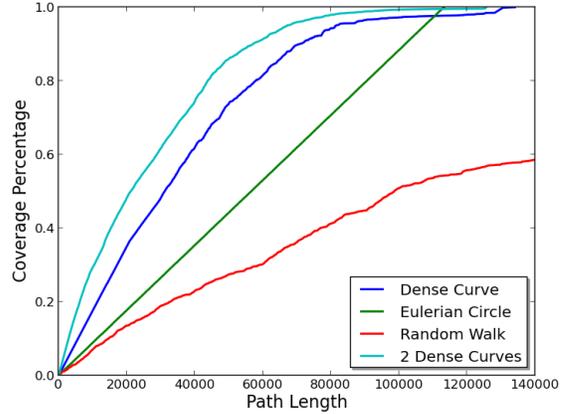}\\
\end{tabular}
\caption{Coverage rate under different path budget. \label{fig:coverageFig}}
\end{center}
\vspace*{-4mm}
\end{figure}

\begin{figure}[h]
\begin{center}
\begin{tabular}{cc}
\includegraphics[width=0.45\textwidth]{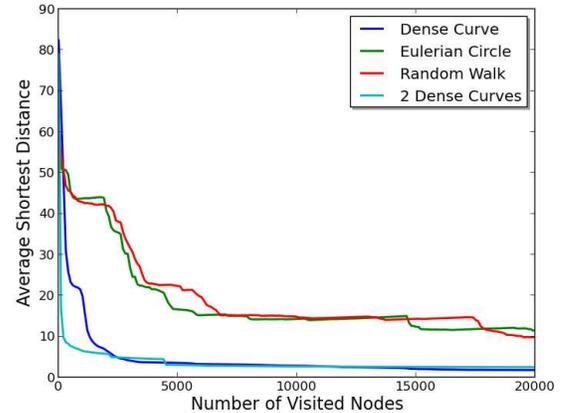}\\
\end{tabular}
\caption{Distance from uncovered nodes to covered nodes. \label{fig:distanceFig}}
\end{center}
\vspace*{-4mm}
\end{figure}

\subsection{The coverage rate under multiple data mules}
Our method can easily generate paths for multiple data mules. 
By deliberately assigning slopes and starting positions of the multiple dense curves, the data mules may cover the network with little overlap between the paths. We can see from Figture~\ref{fig:coverageFig} that two dense curve paths could cover the network even faster than the one curve case. This is not surprising as we have one more path, however, what interesting is that the two paths have little overlap. In fact, there are only $1865$ nodes overlap out of the first $10000$ covered nodes under each path. In Figure~\ref{fig:distanceFig}, two dense curves show better performance especially at the beginning than one dense curve under the average shortest distance. For the first covered $1000$ nodes, the average shortest distance between visited nodes and unvisited nodes under two dense curves is 6.5 while under one dense curve is 20. The multiple paths generated by our method could cover the network more uniformly and faster.

\section{Conclusion}
We show in this paper a new construction for computing a dense curve on a 3D sensor network when the sensors are densely on a 2D manifold. The algorithm substantially generalizes over the prior work by Ban~\etal~\cite{ban13topology} while keeping essentially the same nice properties. As future work we would like to see how to generalize the idea to truly 3D networks (volumetric 3D sensor networks).

\end{document}